\documentclass[lettersize,journal]{IEEEtran}

\usepackage{amsmath,amssymb,amsthm,amsfonts}
\usepackage[caption=false,font=normalsize,labelfont=sf,textfont=sf]{subfig}

\usepackage{textcomp}
\usepackage[linesnumbered,ruled,vlined]{algorithm2e}
\usepackage{url}
\usepackage{verbatim}
\usepackage{graphicx}
\usepackage{cite}
\usepackage{xspace}
\usepackage{xcolor}
\usepackage{mathtools}
\usepackage{hyperref}
\usepackage{enumitem}

\usepackage{svg}
\usepackage{pgfplots}
\usetikzlibrary{decorations.pathmorphing, patterns,shapes}
\usepackage{adjustbox}
\usepackage{tikz, tikz-cd}
\usepackage[all]{xy}
\usetikzlibrary{decorations.pathmorphing, patterns,shapes}
\usetikzlibrary{intersections,calc}
\usepackage{mathtools}
\newtheorem{theorem}{Theorem}[section]

\theoremstyle{definition}  
\newtheorem{definition} [theorem] {Definition}

\DeclareMathOperator{\St}{st}
\DeclareMathOperator{\Lk}{lk}

\newcommand{\Real}{\mathbb{R}}
\newcommand{\Mspace}{\mathbb{M}}
\newcommand{\Xspace}{\mathbb{X}}
\newcommand{\Zspace}{\mathbb{Z}}
\newcommand{\Cerfgraphic}{\mathcal{C}}

\newcommand{\localECC}[2]{\chi_{#1}(#2)}
\newcommand{\ecc}[1]{\mathcal{E}_{#1}}
\newcommand{\TVECC}[1]{\mathcal{E}_{\{#1\}}}
\newcommand{\orecc}{\operatorname{ECC}}
\newcommand{\localTVECC}[1]{\chi_{#1}}
\newcommand{\ie}{i.e.,\xspace}

\newcommand{\inlineheading}[1]{%
  \vspace{1ex}%
  \noindent\textbf{#1}%
}

\newcommand{\revision}[1]{#1}
\newcommand{\revisionsecond}[1]{{#1}}
\newcommand{\cancel}[1]{}

\begin{document}

\title{PL Morse-Cerf Theory for Topological Analysis of Time-Varying Scalar Fields}

\author{Amritendu Dhar, Apratim Chakraborty, and Vijay Natarajan
\thanks{Amritendu Dhar and Vijay Natarajan are with the Indian Institute of Science, Bangalore. \{amritendud,vijayn\}@iisc.ac.in. Apratim Chakraborty is with TCG CREST, Kolkata, India. apratimn@gmail.com.
}%
\thanks{Manuscript received May xx, 202x; revised August xx, 202x.}
}

\markboth{IEEE Transactions on Visualization and Computer Graphics,~Vol.~X, No.~X, Month~Year}%
{Dhar \MakeLowercase{\textit{et al.}}: PL Morse-Cerf Theory for Topological Analysis of Time-Varying Scalar Fields}


\maketitle

\begin{abstract}
Morse-Cerf theory studies the evolution of critical points in one-parameter families of smooth functions defined on a manifold. This paper describes an adaptation of Morse-Cerf theory of smooth functions to one-parameter families of piecewise-linear~(PL) functions. The vertex diagram and Cerf diagram are introduced as topological descriptors that represent evolution of the PL function. The characterization of a crossing in the vertex diagram based on the homology of the lower links of vertices leads to the definition of a topological descriptor for time-varying scalar fields. Algorithms and methods for computing, querying, and comparing Cerf diagrams support its application towards the analysis of time-varying scalar fields. Stability theorems on the comparison measure provide a guarantee of robustness. Experimental results on multiple scientific datasets demonstrate its practical utility. 
\end{abstract} 


\section{Introduction}
\IEEEPARstart{T}{opological} analysis of a scalar field typically proceeds by viewing the field as a Morse function and constructs Morse-theoretic descriptors such as the Morse-Smale complex, contour tree, merge tree, Reeb graph~\cite{heine2016survey,yan2021scalar}. A sequence of elementary moves, including the cancellation or creation of a pair of critical points, can be used to transform one Morse function into another. These elementary moves applied on a topological descriptor enable the development of algorithms for simplification, comparison, and feature identification. Cerf theory is a natural extension of this approach to one-parameter family of smooth functions~\cite{cerf1970stratification}. A key construct in this theory is the Cerf graphic of a generic one-parameter family of smooth functions, which tracks the critical values as the parameter varies. This theory however, is developed in the context of smooth functions and a PL analog is not yet available. Such a transportation of ideas from Cerf theory to a family of PL functions is necessary in order to apply these ideas towards the analysis of time-varying scalar fields or other one-parameter families such as ensemble data. The study and analysis of time-varying scalar fields is pivotal in understanding complex dynamical systems. We present the first step towards the development of a PL analog of Morse-Cerf theory with a focus on time-varying scalar field analysis.

\subsection{Related work}
In the context of time-varying scalar fields, previous work has described methods for computing the evolution of topological descriptors with a focus on characterizing changes to the combinatorial representation, including time-varying Reeb graph~\cite{edelsbrunner2008timevaryingreeb}, merge tree~\cite{oestering2017computing,saikia2017global,kopp2023temporal},  and extremum graph~\cite{das2024tveg}. Additionally, data structures such as the nested tracking graph~\cite{lukasczyk2017nested,Lukasczyk2020NestedTrackingGraph} and temporal treemaps~\cite{kopp2019temporal} that capture temporal evolution of trees have been applied to time-varying scalar fields. 
Comparison measures between topological descriptors~\cite{saikia2014extended,soler2018lifted,sridharamurthy2020,pont2022wasserstein,yan2021scalar,Yan2023} have been used to directly track critical points over time.
Other approaches study the evolution of persistence diagrams, including vineyards~\cite{cohen2006vines} and multi-parameter persistence modules arising from a one-parameter family\cite{Bubenik2024}.
While these descriptors do support the tracking of critical points of a scalar field, a descriptor that provides a comprehensive representation of all topological events for PL functions and the supporting theoretical foundation is not yet available.

Morse theory~\cite{matsumoto2002} describes how a Morse function on a smooth $d$-dimensional manifold induces a handle decomposition of the manifold. A handle decomposition is a process of constructing the manifold by gluing together simple building blocks, called handles, of each lower dimension. The handle decomposition associated with a Morse function plays an important role in the computation of topological descriptors such as the Morse-Smale complex~\cite{EHZ03,EHNP03,defloriani2015morse}. PL and discrete analogs of Morse theory~\cite{forman2002user,gyulassy2006simplification,robins2011theory,gunther2012efficient,shivashankar3d} are especially amenable to computation and have found wide applicability in topological data analysis~\cite{heine2016survey}.

Cerf proved that a generic one‑parameter family of smooth functions $\{f_t\}_{t\in[0,1]}$ is Morse for all but finitely many parameter values~\cite{cerf1970stratification}. Within a neighborhood of one such parameter value, the handle decomposition changes by exactly one of two elementary events: (a)~Handle creation or annihilation, which corresponds to the birth or death of a pair of critical points, and (b)~Handle slide, where one handle slides over another of the same index without altering the set of indices.

The \emph{Cerf decomposition} is a sequence of handle decompositions linked by the above-mentioned elementary events. Analogous to the classical Morse setting, a Cerf decomposition forms the basis for the construction of topological descriptors of a time‑varying scalar field. This motivates the study of time‑varying scalar fields within the framework of Cerf theory, a natural extension of the Morse‑theoretic approach to the study and analysis of scalar fields. The proposed theoretical development is in contrast to existing methods mentioned above that focus on the development of novel methods for individual tasks such as tracking, temporal summarization, or compression.

\subsection{Contributions}
We present a piecewise-linear (PL) adaptation of Cerf theory that enables the study and topological analysis of time-varying scalar fields defined on a combinatorial manifold. 
\begin{itemize}[leftmargin=*]
    \item[1.] \textit{Topological descriptors}: Central to the development of the adaptation is the introduction of two descriptors, the vertex diagram and the Cerf diagram, that support the study of critical points of a time-varying field. A vertex diagram depicts the evolution of scalar values of vertices of a time-varying scalar field on a combinatorial manifold $\Mspace$ together with the descriptive homological indices. A PL Cerf diagram, derived from a vertex diagram, represents the evolution of the critical points.
    \item[2.] \textit{Structural characterization}: We study  degeneracies in the time-varying scalar field by characterizing crossings in the vertex diagram in terms of  changes to the time-varying Betti number $\beta_{i}^{t}(v)$ of the link neighborhood of a critical vertex~$v$. 
    \item[3.] \textit{Computation}: We also describe a simple algorithm for efficient construction of the Cerf diagram. 
    \item[4.] \textit{TV-ECC: Comparison and Stability}: Next, we introduce a topological descriptor of a time-varying scalar field, called the time-varying Euler Characteristic Curve (TV-ECC) and its local variant, local TV-ECC, defined at each vertex. This descriptor aggregates the evolution of the time-varying Betti number and is used to define a comparison measure between two Cerf diagrams, which essentially aggregates the differences between critical values over time. Additionally, we present a proof of stability of TV-ECC.
    \item[5.] \textit{Computational experiments}: Finally, we present experiments on multiple datasets where the Cerf diagram and the associated comparison measure support the visual analysis of time-varying scalar fields, thereby demonstrating its utility.
\end{itemize} 
The Cerf diagram and vertex diagram together with preliminary results were described in a short paper~\cite{dhar2025analyzing}. This extended paper additionally presents new theoretical results characterizing the TV-ECC, a stability theorem, the notion of Cerf tracks, computational experiments, and a comparative discussion regarding topological descriptors. Further, this paper also expands on the exposition of the PL adaptation, algorithms, and experiments.



     
    

\section{Preliminaries}
\label{sec:preliminaries}
We introduce the necessary terms and concepts from Morse theory for PL functions and Cerf theory for smooth functions~\cite{edelsbrunner2010computational, dey2022computational}. A concise introduction to the fundamental concepts of Cerf theory can be found in Gay and Kirby's paper \cite{gaykirby}. We refer the reader to classic texts and surveys on smooth Morse theory~\cite{milnor1963morse,matsumoto2002} and Cerf theory~\cite{gay2012connected,hatcher1973pseudoisotopies} for a detailed exposition.

\subsection{PL functions on combinatorial manifolds}
A simplicial complex $K$ is called a \emph{triangulation} of a manifold $\Mspace^d$ if the underlying space of $K$ is homeomorphic to $\Mspace^d$. A PL function $f:\Mspace^d \to \Real$ is defined by assigning real values to vertices of the triangulation $K$ and extending the map linearly in the interior of each simplex. 
The \emph{star} of a vertex $u \in K$ is defined as $\St(u)=\{\sigma\in K \mid u\in\sigma\}$. 
The \emph{link} of a vertex $u$ is defined as $\Lk(u)=\{\tau\in K \mid \tau\subset\sigma \mbox{ for some } \sigma\in\St(u),u\notin\tau\}$. 
The \emph{lower link} of a vertex $u$ is a subset of the link and is defined as $\Lk^-(u)=\{\tau\in\Lk(u) \mid f(v)\le f(u) \mbox{ for all } v\in\tau\}$.

A simplicial complex of dimension $d$ is called a \emph{combinatorial manifold} if the link of each vertex is homeomorphic to a sphere of dimension $d-1$. 
Throughout this paper, we work in the category of combinatorial manifolds and analyze topological properties of one‑parameter families of PL functions on them. 
 
The PL function $f$ is called \emph{generic PL} if $f(u) \neq f(v)$ for distinct vertices $u,v \in K$. Vertices of $K$ may be classified as regular or critical with respect to a generic PL function $f$ based on the reduced homology of the lower link~\cite{edelsbrunner2010computational}.  \autoref{fig:CriticalPts-2Manifold} shows an example. They are also referred to as \textit{homologically regular} (\textit{H-regular}) and \textit{homologically critical} (\textit{H-critical})~\cite[Definition 3.3]{Grun-Kuhn-PLMorse}.
\begin{definition}[Regular point]\label{defn:regularpt}
A vertex $v$ is called a \emph{regular point} of $f$ if for every $i \in [0,d]$, $\widetilde H_{i-1}\bigl(\Lk^-(v)\bigr)\;=\;0$.
\end{definition}
\begin{definition}[Critical point]\label{defn:criticalpt}
A vertex $v$ is called a \emph{critical point} of $f$ if the homological index
$\beta(v) =  (\beta_0(v),\beta_1(v),\dots,\beta_d(v))$, where $\beta_i(v) =\dim\widetilde H_{i-1}\bigl(\Lk^-(v)\bigr)$, is not identically zero. 
\end{definition}
In the above, we adopt the convention that the reduced homology group $\widetilde H_{-1}(\Xspace)= 0$ if $\Xspace\neq\emptyset$, and $\widetilde H_{-1}(\Xspace)= \Zspace_2$ if $\Xspace=\emptyset$.
A critical point $v$ is called \emph{non‑degenerate} if exactly one entry $\beta_i(v)$ of its index $\beta(v)$ is non-zero and equals 1. In this case, $v$ is called a critical point of index $i$. If exactly one entry $\beta_i$ is non-zero then the critical point of index $i$ is said to have multiplicity  $\beta_i$. 
\begin{definition}[PL Morse function]\label{defn:plMorse}
A generic PL function whose critical points are all non-degenerate is called a \emph{PL Morse} function.  
\end{definition}

\begin{figure*}[!htb]
\centering
\setlength{\tabcolsep}{6pt}
\begin{tabular}{cccc}
    \includegraphics[width=0.12\textwidth]{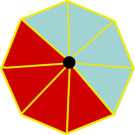} &
    \includegraphics[width=0.12\textwidth]{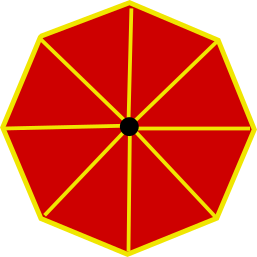} &
    \includegraphics[width=0.12\textwidth]{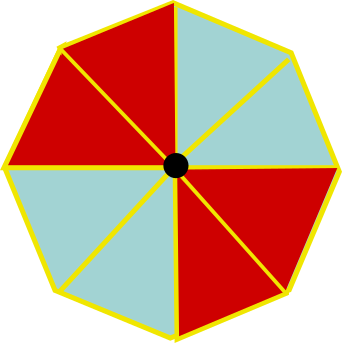} &
    \includegraphics[width=0.12\textwidth]{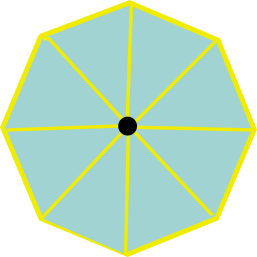}    \\
\small
   Regular, $\beta = (0,0,0)$ &
   Minimum, $\beta = (1,0,0)$ &
   1-Saddle, $\beta = (0,1,0)$ &
   Maximum, $\beta = (0,0,1)$
\end{tabular}

\caption{Critical points of a PL Morse function defined on a combinatorial 2-manifold and their homological index $\beta$. The lower star of the vertex (black) is shown in blue.
}
\label{fig:CriticalPts-2Manifold}
\end{figure*}

\subsection{Cerf theory for smooth functions}\label{sec:Cerfgraphic}
A generic one-parameter family of smooth functions $f_t: \Mspace^d \rightarrow \Real$, where $t$ ranges over an interval $I$, is Morse at all but finitely many time instances. At a degenerate time instance $t_0$, two critical points of $f_t$ may have identical values. As time $t$ increases and crosses $t_0$, the corresponding pair of critical values may cross each other. Alternatively, we observe a birth/death transition of critical points. 
The function $f_t$ can be expressed as follows in terms of local coordinates $(x_1,\ldots,x_d)$ of $\Mspace^d$ within the neighborhood of a degenerate birth/death critical point~\cite{cerf1970stratification}.
{
\setlength{\abovedisplayskip}{0pt}
\setlength{\belowdisplayskip}{0pt}
\begin{equation}
    f_t(x)=x_{k+1}^3 \pm t \revision{ x_{k+1}}-\sum_{i=1}^k x_i^2+\sum_{i=k+2}^d x_i^2.
\end{equation}
}
The cubic term in $x_{k+1}$ unfolds the degenerate quadratic so that at $t=0$ a canceling pair of index-$k$ and index-$(k+1)$ critical points is created or annihilated, as determined by the $\pm$ sign before $x_{k+1}$.
Suppose that each $f_t$ takes values in a compact interval $J \subset \Real$ for $t \in I$. Cerf studied the stratification of the space of smooth functions on $\Mspace^d\times I$~\cite{cerf1970stratification}.  
\begin{definition}[Cerf graphic]\label{defn:smoothcerfgraphic}
The \emph{Cerf graphic} of a smooth family of functions $\{f_t\}$ is a subset of $I\times J$,
\[
  \Cerfgraphic = \bigl\{(t,c)\in I\times J \mid c\text{ is a critical value of }f_t\bigr\}.
\]
\end{definition}
\begin{figure}[ht]
\centering
\includegraphics[width=0.4\textwidth] {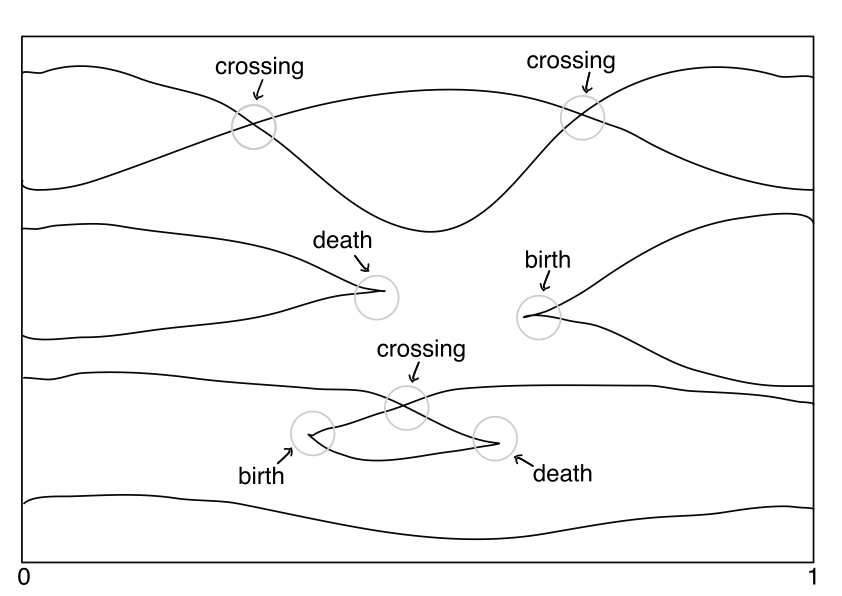}
\caption{Cerf graphic of a one-parameter family of smooth functions.}
\label{fig:smootcerfgraphics}
\end{figure}
As illustrated in \autoref{fig:smootcerfgraphics}, the Cerf graphic tracks the critical values of a smooth family of functions as the parameter $t$ varies from 0 to 1. The cusps, representing the birth/death of a pair of critical points, correspond to the creation/annihilation of a pair of handles. Crossings of paths of two critical points correspond to a handle slide.

\section{PL Morse-Cerf theory}
PL Morse theory analyzes the topology of a combinatorial manifold using PL Morse functions. We introduce an extension of Cerf theory to the PL category aimed at constructing topological descriptors of PL time-varying scalar fields and developing methods for analyzing such fields.

\begin{figure}[ht]
\centering
\begin{tikzpicture}[scale=0.7]
	
		\node  (0) at (-18, 7) {};
		\node  (1) at (-18, 6) {};
		\node  (2) at (-18, 5) {};
		\node  (7) at (-17, 7) {};
		\node  (8) at (-17, 6) {};
		\node  (9) at (-17, 5) {};
		\node  (14) at (-16, 7) {};
		\node  (15) at (-16, 6) {};
		\node  (16) at (-16, 5) {};
		\node  (21) at (-15, 7) {};
		\node  (22) at (-15, 6) {};
		\node  (23) at (-15, 5) {};
		\node  (28) at (-14, 7) {};
		\node  (29) at (-14, 6) {};
		\node  (30) at (-14, 5) {};
		\node  (35) at (-13, 7) {};
	    \node  (36) at (-13, 6) {};
        \node  (91) at (-12.5, 6.5) {};
		\node  (37) at (-13, 5) {};
		\node  (42) at (-12, 7) {};
		\node  (44) at (-12, 5) {};
		\node  (49) at (-11, 7) {};
		\node  (50) at (-11, 6) {};
		\node  (51) at (-11, 5) {};
		\node  (56) at (-10, 7) {};
		\node  (57) at (-10, 6) {};
		\node  (58) at (-10, 5) {};
		\node  (63) at (-9, 7) {};
		\node  (64) at (-9, 6) {};
		\node  (65) at (-9, 5) {};
		\node  (70) at (-19, 7) {Vertex $v_3$};
		\node  (71) at (-19, 6) {Vertex $v_2$};
		\node  (72) at (-19, 5) {Vertex $v_1$};
		\node  (80) at (-18, 4) {$0$};
		\node  (81) at (-12, 6) {};
    
		\node  (82) at (-17, 4) {$0.2$};
		\node  (83) at (-16, 4) {$0.3$};
		\node  (85) at (-13, 4) {$0.6$};
		\node  (87) at (-12, 4) {$0.7$};
	
		\node  (89) at (-9, 4) {$1$};

    	\node  (90) at (-19, 4) {$t$};
	
		\draw[dotted] (50.center) to (64.center);
		\draw[dotted] (36.center) to (91.center);
		\draw (42.center) to (63.center);
		\draw[dotted]  (91.center) to (81.center);
		\draw[dotted]  (81.center) to (50.center);
		\draw (2.center) to (14.center);
		\draw[dotted] (1.center) to (7.center);
		\draw (0.center) to (16.center);
		\draw (16.center) to (65.center);
		\draw (14.center) to (35.center);
            \draw (35.center) to (91.center);
            \draw (91.center) to (42.center);
		\draw[dotted] (7.center) to (15.center);
		\draw[dotted] (15.center) to (36.center);
\end{tikzpicture}
\caption{Vertex diagram of a PL time-varying function represents the evolution of regular (dotted) and critical points (solid).} 
\label{fig:vertexdiag}
\end{figure}

\inlineheading{Diagrams.}
Let $\Mspace$ be a $d$-dimensional combinatorial manifold and $V$ be the vertex set of its triangulation $K$. We will consider the  time interval $I = [0,1]$. A one-parameter PL family $\left\{f_t\right\}$ is called generic if the collection of curves $\left\{\left(t, f_t\left(v\right)\right) \mid t \in[0,1]\right\}$ for all $v\in V$ contains finitely many double points at distinct time steps $0<t_1<t_2<\cdots<t_m<1$ and contains no higher degree intersection point. 
\begin{definition}[Vertex diagram]
The \emph{vertex diagram} of a generic one-parameter PL family $\{f_t\}$ is given by the set $\left\{\left(t, f_t\left(v\right)\right) \mid t \in[0,1], v \in V\right\}$.
\end{definition}
\autoref{fig:vertexdiag} shows a vertex diagram, consisting of three vertices, that distinguishes between critical and regular points. Next, we introduce the notion of a PL Morse family. 
\begin{definition}[One-parameter PL Morse family]
A one-parameter PL family $\left\{f_t\right\}$ is called a \emph{PL Morse family} if in addition to being generic, $f_t$ is PL Morse for $t \neq t_1, \ldots, t_m$.
\end{definition}
Note that being PL Morse is not a generic property of a one‑parameter PL family. However, in practice any such family can be approximated by a PL Morse family. Hence, theoretical results derived for PL Morse families are useful for analyzing time-varying scalar fields.
\begin{definition}[Cerf diagram]
Let $C_t\subset V$ be the set of critical points of $f_t$. The \emph{Cerf diagram} of the family $\{f_t\}$ is given by the set $\left\{(t,f_t(v)) \mid \ t\in[0,1], v\in C_t\right\}$.
\end{definition}
\autoref{fig:cerf_2Gaussians} shows the Cerf diagram of a synthetic sum-of-Gaussians dataset. We denote by $\Lk_t^-(v)$ the lower link of the vertex $v$ with respect to $f_t$, and define time-varying Betti numbers as $ \beta_i^t(v)\;:=\;\dim\widetilde H_{i-1}\bigl(\Lk_t^-(v)\bigr)$.

\begin{figure*}[!htb]
\centering
\includegraphics[width=0.65\linewidth]{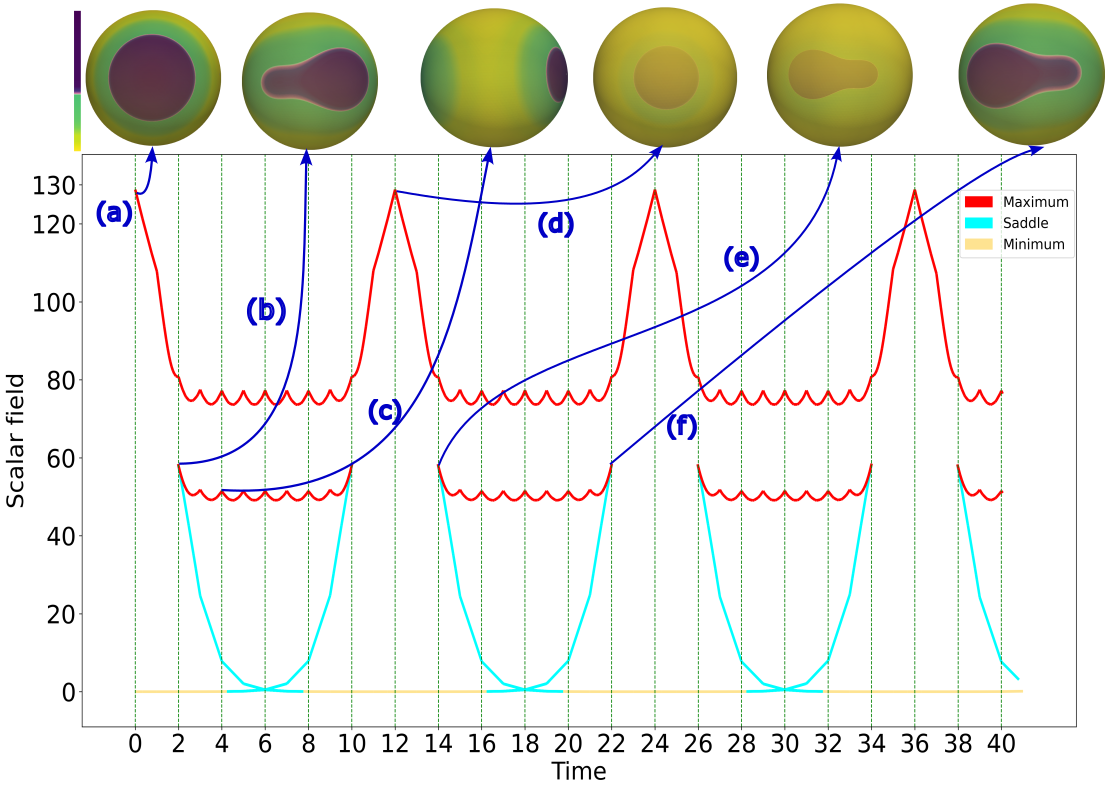}
\caption{Cerf diagram of a synthetic sum-of-Gaussians scalar field, where the centres of the two Gaussians move in opposite directions along the equator of a sphere at a speed of $\pi/12$. We observe a (half) period of 12 from the Cerf diagram because the centres meet at the opposite pole.\textbf{(a)} The centres split at $t=0$, \textbf{(b)} a second maximum appears at $t = 2$ resulting in a birth event. \textbf{(c)} Both maxima are visible at $t=4$. \textbf{(d)} at $t = 12$ the Cerf diagram records a peak when the centers meet  diametrically opposite to their starting location. \textbf{(e)} Another birth event is seen as the maxima split again at $t = 14$.
\textbf{(f)} This second maximum dies at $t=22$ just prior to the maxima returning to their starting location.
}
\label{fig:cerf_2Gaussians}
\end{figure*}

\inlineheading{Crossings.} 
The vertex diagram encodes sufficient information to derive the tracking graph of critical points, and hence enables the tracking of the associated topological features. We present an exhaustive list of the types of crossings between two time instances $t_1$ and $t_2$ in a vertex diagram of a PL Morse family. In each case, we also describe the effect on the tracking graph within the time interval $[t_1,t_2]$.

\textit{Critical–critical crossing}: Two critical points cross without a change in their indices; no corresponding vertex movement occurs in the tracking graph.

\textit{Regular–regular crossing}: Two regular points cross and remain regular; no corresponding vertex movement in the tracking graph.

\textit{Critical–regular crossing}: A critical point and a regular point cross following which neither transitions between critical and regular; no corresponding vertex movement in the tracking graph.

\textit{Critical–regular switch}: A critical point and a regular point cross, then both transition; in the tracking graph, the feature at the critical point before the crossing moves linearly toward the second vertex.

\textit{Critical–critical index swap}: Two critical points cross and remain critical but their Morse indices interchange; in the tracking graph, the feature at each vertex moves linearly toward the other vertex.

\textit{Birth crossing}: Two regular points cross and both transition to become critical; in the tracking graph, two new features are born at time of crossing at the two vertices.

\textit{Death crossing}: Two critical points cross and both become regular; in the tracking graph, the two corresponding features die at the time of crossing.


%
%
The Betti number $\beta_{i}^{t}(v)$ of the lower link of a vertex $v$ changes after a crossing. The following theorem characterizes this change.

\begin{theorem}\label{thm:betti-update}
Let $\left\{f_t\right\}$ be a generic one-parameter PL family. If the vertex diagram of $\left\{f_t\right\}$ contains exactly one crossing in a given time range $[t_1,t_2]$, namely between two vertices $v$ and $u$, then the Betti numbers of the lower link of $v$ and $u$ satisfy the following relationship:
{
\setlength{\abovedisplayskip}{0pt}
\setlength{\belowdisplayskip}{0pt}
   \[  \sum_{i=0}^{d} {(-1)}^{i} (\beta_{i}^{t_2}(v) - \beta_{i}^{t_1}(v)) \;=\; - \sum_{i=0}^{d} {(-1)}^{i} (\beta_{i}^{t_2}(u) - \beta_{i}^{t_1}(u)) \]
}
\end{theorem}

\begin{proof}
Let $I_{uv}$ denote the intersection between the lower links of vertices $v$ and $u$ at time $t_1$, $I_{uv} \;=\;\Lk_{t_1}^-(v)\,\cap\,\Lk_{t_1}^-(u)$. Without loss of generality, assume that $f_{t_1}(v) > f_{t_1}(u)$. If vertices $v$ and $u$ are connected by an edge of $K$, then at a time instance $t_2$ after the crossing, their lower links may be expressed as follows:
\begin{align*}
  \Lk_{t_2}^-(u)
    &= \Lk_{t_1}^-(u)\;\cup\;\bigl(v\star I_{uv}\bigr),\\
  \Lk_{t_2}^-(v)
    &= \Lk_{t_1}^-(v)\;\setminus\;\bigl(u\star I_{uv}\bigr).
\end{align*}
Applying the Mayer–Vietoris sequence to the decomposition of $\Lk_{t_2}^-(u)$ yields the long exact sequence
\[
\begin{aligned}
  \cdots\;&\to\;\widetilde H_i(I_{uv})
  \;\to\;\widetilde H_i\bigl(\Lk_{t_1}^-(u)\bigr)\oplus\widetilde H_i(v\star I_{uv})\\
  &\to\;\widetilde H_i\bigl(\Lk_{t_2}^-(u)\bigr)
  \;\to\;\widetilde H_{i-1}(I_{uv})
  \;\to\;\cdots
\end{aligned}
\]
for integer $i\in[0,d]$. Similarly, we obtain a long exact sequence for $\Lk_{t_2}^-(v)=\Lk_{t_1}^-(v)\setminus(u\star I_{uv})$.
Since $v\star I_{uv}$ and $u\star I_{uv}$ are contractible, their reduced homology groups vanish, and the above exact sequences simplify to
\[ 
\begin{aligned}
\cdots\;\to\; \widetilde H_i(I_{uv})
  \;\to\; \widetilde H_i\bigl(\Lk_{t_1}^-(u)\bigr)  \;\to\; \widetilde H_i\bigl(\Lk_{t_2}^-(u)\bigr)
  \;\to\; \cdots,
\end{aligned}
\]
\[
\begin{aligned}
  \cdots\;\to\;
  \widetilde H_i(I_{uv})
  \;\to\;  \widetilde H_i\bigl(\Lk_{t_2}^-(v)\bigr)
  \;\to\;  \widetilde H_i\bigl(\Lk_{t_1}^-(v)\bigr)
  \;\to\;  \cdots.
\end{aligned}
\]
Let $r_i$ denote the rank of $\widetilde H_{i-1}(I_{uv})$.  Counting alternating sums of ranks in both exact sequences, we obtain, 
{
\setlength{\abovedisplayskip}{0pt}
\setlength{\belowdisplayskip}{0.5pt}
\begin{align*}
  \sum_{i=0}^{d}(-1)^i r_i
  \;+\;
  \sum_{i=0}^{d}(-1)^{\,i+1}\beta_i^{t_1}(v)
  \;+\;
  \sum_{i=0}^{d}(-1)^i\beta_i^{t_2}(v)
  \;&=\;0,\\
  \sum_{i=0}^{d}(-1)^i r_i
  \;+\;
  \sum_{i=0}^{d}(-1)^{\,i+1}\beta_i^{t_1}(u)
  \;+\;
  \sum_{i=0}^{d}(-1)^i\beta_i^{t_2}(u)
  \;&=\;0.
\end{align*}
}
Equating the two expressions, removing the common term involving $r_i$, and combining the remaining two terms, we obtain the desired relationship.
\end{proof}

\inlineheading{Topological descriptor.}
The evolution of the Betti numbers characterize crossings in the vertex diagram. We aggregate the evolution over all critical points to develop a topological descriptor of a one-parameter PL family. 
\begin{definition}[Local TV-ECC]\label{defin:tvecc-local}
Given a generic one-parameter PL family $\{f_t\}$ on a combinatorial $d$-manifold $\Mspace$ and a vertex $v \in V$, define the \emph{local time-varying Euler characteristic curve} $\localTVECC{v}: \Real \times [0,1] \to \Real$ as 
\[
\localTVECC{v}(s,t) := 
\begin{cases}
\displaystyle\sum_{i=0}^{d}(-1)^i\,\beta_i^{t}(v) & \text{if } f_t(v) \leq s, \\
0 & \text{otherwise.}
\end{cases}
\]
\end{definition}
Dlotko and Gurnari~\cite{ecc} introduced the Euler Characteristic Curve ($\orecc$) as a shape invariant on a filtered simplicial complex, and demonstrated its use for topological data analysis.
For a single parameter filtration $K_s$ of the simplicial complex $K$, $\orecc$ is defined as $\orecc(K,s) = \chi(K_s)$.
For a filtration induced by a PL function $f$, the $\orecc$ will be denoted by $\ecc{f}$. The  following theorem shows that $\orecc$ can be expressed as a sum of $\localTVECC{v}$ over all vertices.

\begin{theorem}\label{thm:localsumthm}
Let $\Mspace$ be a combinatorial $d$-manifold and $f \colon \Mspace \to \Real$ a generic PL function whose image lies in a compact interval~$J\subset\Real$. Then, for every $s\in\Real$,
\[
  \ecc{f}(s)\;=\;
  \sum_{v\in V} \localECC{v}{s}.
\]
\end{theorem}

\begin{proof}
Label the vertices of $\Mspace$ by $v_1,\dots,v_n$ such that 
$f(v_1)<f(v_2)<\dots<f(v_n)$.
Let $R:=\max\{i\mid f(v_i)<s,s\in J\}$. Choose numbers $s_0,s_1,\dots,s_R \in \Real$ satisfying 
\[ 
\begin{aligned}
 s_{i-1}<f(v_i)<s_i, \,(s_{i-1},s_i)\cap f(V)=\{f(v_i)\}\\
 \text {for } 1\le i < R, \text { and }  s_R \ge f(v_R).   
\end{aligned}
\]
Define $A_i:=f^{-1}(s_{i-1}),  J_i:=[s_{i-1},s_i]$, and $B_i:=f^{-1}(J_i)$ for $1\le i\le R$. Notice that $A_0=\varnothing$. 
By the analog of fundamental theorem of Morse theory in PL case~\cite[Proposition 2.7]{Bestvina2008PLMorse}, there is a homotopy
equivalence relative to $A_i$ (\ie $A_i$ remains fixed), 
\[
  (B_i,A_i)\;\simeq\;(A_i\cup C_i,A_i),
\]
where $C_i$ is the cone on $\Lk^{-}{v_i}$, with cone point $v_i$ and base  $A_i$. The suspension $\Sigma\,\Lk^{-}(v_i)$ is formed by connecting $\Lk^{-}(v_i)$ to two cone points. For each integer $k\in[0,d]$,
\begin{align*}
H_{k}(A_i\cup C_i,\,A_i)
&\cong \widetilde H_{k}\bigl((A_i\cup C_i)/A_i\bigr)
    &&\text{\footnotesize(excision)}\\
&\cong \widetilde H_{k}\bigl(\Sigma\,\Lk^{-}(v_i)\bigr)
    &&\text{\footnotesize(cone quotient)}\\
&\cong \widetilde H_{k-1}\bigl(\Lk^{-}(v_i)\bigr).
    &&\text{\footnotesize(suspension isomorphism)}
\end{align*}
Therefore, 
\[
  H_k(B_i,A_i)\;\cong\;
  \widetilde H_{k-1}\!\bigl( \Lk^{-}(v_i)\bigr)
  \quad\text{for all }k\ge 0,
\]
so that
\[
  \chi(B_i,A_i)=\sum_{k=0}^{d}(-1)^k
               \widetilde{\beta}_{k-1}\!\bigl(\Lk^{-}(v_i)\bigr)
               =\localECC{v_i}{f(v_i)}.
\]
Expressing $\chi(B_R)$ as a sum of the relative Euler characteristics $\chi(B_i,A_i)$, we obtain
\[
  \ecc{f}(s)=\chi(B_R)=\sum_{i=1}^{R}\chi(B_i,A_i)
          =\sum_{i=1}^{R}\localECC{v_i}{f(v_i)}
          =\sum_{v\in V}\localECC{v}{s}.
\]

\end{proof}

\begin{definition}[TV-ECC]
Given a generic one-parameter PL family $\left\{f_t\right\}$, the \emph{time-varying Euler Characteristic Curve} $\TVECC{f_t}:\Real \times [0,1] \to \Real$ is defined as 
 \[ \TVECC{f_t}(s,t):= \sum_{v \in V} \localTVECC{v}(s,t). \]
\end{definition}

\section{Computing Cerf diagrams}
\label{sec:algorithm}
In this section, we describe an algorithm to compute the Cerf diagram and subsequently discuss methods for querying the diagram.
A time-varying scalar field $f_t$ is available as a sequence $\{f_t(v)\}$ of $T$ scalar values at each vertex $v \in V$. In the following description, we further assume that the scalar value at vertices are distinct within a time step. The scalar values at a vertex are linearly interpolated between every pair of consecutive time steps resulting in a time-varying PL scalar field. 
\revision{In many real-world scenarios, the input data may be a discrete time sample of a general continuous time-varying field. In the absence of prior knowledge or analytic description of the underlying time-varying field, a piecewise-linear interpolant seems a natural option. The development of a PL analog of Morse theory towards scalar field analysis~\cite{EHZ03,heine2016survey,yan2021scalar} also followed from a similar reasoning regarding the suitability of a linear interpolation.}

We now describe \autoref{alg:cerfDiagram}, which computes the Cerf diagram by tracking all crossings in the vertex diagram between two vertices that lie in the link neighborhood of each other. Let $S_v$ denote the path corresponding to the vertex $v$ in the vertex diagram, stored as a sequence of edges. In \autoref{alg:Cerf-intersections}, we compute all intersections between the paths of two vertices $u$ and $v$. The set $\mathcal{I}_v$ records all instances $t$ when a vertex $u \in \Lk(v)$ enters the lower link of $v$, resulting in an intersection between $S_v$ and another path in the vertex diagram. 
In \autoref{alg:cerf-llinkVertices}, the set $\Lk^{(0)-}(v)$ of vertices lying in the lower link of $v$ at the current time step $t$ is computed and stored. Note that the lower link of $v$, $\Lk^{-}(v)$ is the subcomplex induced by $\Lk^{(0)-}(v)$.
While processing a crossing of vertices $u$ and $v$ (\autoref{alg:cerf-processcrossing}), the algorithm updates $\Lk^-(u)$ and $\Lk^-(v)$ and recomputes $\beta(u)$ and $\beta(v)$ to determine the criticality of the vertices. If the vertices are critical, the correspond arc between the vertices in the vertex diagram is declared as a \textit{Cerf arc} (\autoref{alg:cerf-cerfarc}). The Cerf diagram is represented as a collection of 6-tuples $\{(t_1, f_{t_1}(v),t_2,f_{t_2}(v),v,\beta(v))\}$ where $(t_1, f_{t_1}(v))$ and $(t_2,f_{t_2}(v))$ are the end points of a Cerf arc, $v$ represents the vertex location of the critical point in the time interval \revision{$[t_1,t_2)$}, and $\beta(v)$ is its (fixed) homological index within the time interval \revision{$(t_1,t_2)$}. 
The algorithm handles degenerate critical points since it records the homological index, and hence works in practice for time-varying functions that are not necessarily PL Morse.
\begin{algorithm}[!ht]
\caption{Compute Cerf Diagram}
\label{alg:cerfDiagram}
 \SetKwInOut{Input}{Input}\SetKwInOut{Output}{Output}

 \Input{$K; \{f_t(v)\}$ for $t \in \{0,1,\ldots,T-1\}, v \in V$} 
 \Output{$C$, set of Cerf arcs}

\SetKwFunction{computecerfdiagram}{\textsc{compute-cerf-diagram}}
\SetKwProg{myproc}{Procedure}{}{}
\BlankLine 

\myproc{\computecerfdiagram{$K, \left\{f_t\right\}$}}{

{$\mathcal{I}  \gets \emptyset$, $C \gets \emptyset$}
\BlankLine
\For{$v \in V$ and $t \in \{0,1,\ldots,T-2\}$}
{
  {$S_v \gets \{ ((t,f_t(v)), (t+1, f_{t+1}(v)) ) \}$\\ \tcp{Path of $v$ in vertex diagram}
  \label{alg:cerf-storeSv}
  } 
  \BlankLine 
  {$\mathcal{I}_v \gets \{(t', u, v, f_{t'}(v)) \mid  u \in \Lk^{-}_{t+1}(v), u \notin \Lk^{-}_{t}(v), f_{t'}(u) = f_{t'}(v) \mbox{ for } t< t' < t+1\}$ \tcp{Crossing at $t'$, when $u$ enters $\Lk^{-}_{t+1}(v)$.}
  \label{alg:Cerf-intersections}
  }
  \BlankLine
  {$\mathcal{I}_v \gets \mathcal{I}_v \cup \{(t,v,v,f_t(v)), (t+1,v,v,f_{t+1}(v))\}$ \tcp{Include samples at $t, t+1$.}
  } 
  {$\mathcal{I}  \gets  \mathcal{I} \cup \mathcal{I}_v $}
  \BlankLine
}

{Sort $\mathcal{I}$ with respect to index $t$ }\\
{ $\Lk^{(0)-}(v) \gets \{ u \in V(\Lk(v)), f_0(u) < f_0(v) \}$ 
\tcp{Vertices in $\Lk^-(v)$}
\label{alg:cerf-llinkVertices}
}
\BlankLine
\For{$(t, u, v, f_t(v)) \in \mathcal{I}, u \neq v$}
{
    \label{alg:cerf-processcrossing}
     {$\Lk^{(0)-}(v) \gets \Lk^{(0)-}(v) \cup \{ u\}$
     }
     {$\Lk^{(0)-}(u) \gets \Lk^{(0)-}(u) \backslash  \{ v\}$\\ \tcp{$u$ enters $\Lk^-(v)$, $v$ exits $\Lk^-(u)$}}
    {Update $\Lk^-(v), \Lk^-(u)$}\\
    \For{$z \in \{u,v\}$}
    {
        \If {$\widetilde{H_*}(\Lk^-(z)) \neq 0$ and next tuple in $\mathcal{I}$ containing $z$ occurs at $t'$}
        {
        {$C \gets C \cup ( t, f_t(z), t', f_{t'}(z),z, \beta(z))$
        \label{alg:cerf-cerfarc}
        }
        }
    }
}

}
\end{algorithm}

\inlineheading{Runtime analysis.}
Let $l$ denote the upper bound on the size of the link of a vertex. So, a vertex participates in at most $O(l)$ crossings between two successive time steps, and total number of crossings between two consecutive time steps is $O(nl)$.
Each crossing triggers a computation of homology of the lower link requires $O(l^\omega)$ time, where $\omega$ is the exponent of matrix multiplication. So, the overall runtime for computing the Cerf diagram is $O(nl^{\omega+1}T)$.

\inlineheading{Querying the Cerf diagram.}
We describe queries that support the discovery of spatio-temporal features of interest from the Cerf diagram. The queries report a collection or cluster of Cerf arcs. A cluster of Cerf arcs, called a \textit{Cerf track}, represents the evolution of a collection of critical points that together form a topological feature.

Cerf arcs may be grouped together based on various clustering criteria. Given a point $(t,f)$ on the Cerf diagram, its $(\Delta t, \Delta f)$ neighborhood is the set of all points $(t',f')$ such that $|t-t'| \leq \Delta t$ and $|f-f'| \leq \Delta f$. A $(\Delta t, \Delta f)$ neighborhood of a Cerf arc $\zeta$ is the set of all points in the Cerf diagram that lie within the $(\Delta t, \Delta f)$ neighborhood of any point on $\zeta$.
Two Cerf arcs $\zeta_1$ and $\zeta_2$ are said to lie within a $(\Delta t, \Delta f, k)$ neighborhood of each other if any point in $\zeta_1$ ($\zeta_2$) lies within a $(\Delta t, \Delta f)$ neighborhood of $\zeta_2$ (respectively, $\zeta_1$) and the length of the path \revision{in the manifold} connecting the corresponding critical vertices is at most $k$.
Define a relation $\sim$ between Cerf arcs such that $\zeta_1 \sim \zeta_2$ if $\zeta_1$ and $\zeta_2$ lie within a $(\Delta t, \Delta f, k)$ neighborhood of each other. The equivalence classes of this relation are computed to cluster a collection of Cerf arcs into Cerf tracks.

The input set of Cerf arcs can be generated from a carefully chosen set of seeds on the Cerf diagram. For instance, all Cerf arcs that belong to a user supplied time and function value range could constitute the seeds. Seeds could also be constituted by critical vertices in a user-specified spatial neighborhood in the domain. The associated query could aid in the investigation of features within the corresponding spatio-temporal region.

We remark that the computation of Cerf tracks in our experiments is limited to maxima because they  often correspond to features of interest. The methodology is, however, applicable to all Cerf arcs. The Cerf track computation is summarized in \autoref{alg:cerfarcs}. We assume that $\Delta t, \Delta f,$ and $k$ are user-defined parameters.

\begin{algorithm}[!ht]
\caption{Compute Cerf Tracks}
\label{alg:cerfarcs}
 \SetKwInOut{Input}{Input}\SetKwInOut{Output}{Output}

 \Input{Cerf diagram $C$; Seeds $S \subset C$; $\Delta f, \Delta t, k$;} 
 \Output{$Q$, set of Cerf tracks}

\SetKwFunction{computecerfarcs}
{\textsc{compute-cerf-tracks}}
\SetKwProg{myproc}{Procedure}{}{}
\BlankLine

\myproc{\computecerfarcs{}}{

{\tcp{$\zeta_1 \sim \zeta_2$ if $\zeta_1$ is in $(\Delta f, \Delta t, k)$ neighborhood of $\zeta_2$}}
\BlankLine
{$Q \gets \emptyset$}\\
\For{$seed \in S, seed \notin Q $}
{
    {$q \gets \{ seed\}$\tcp{New cluster}
    }
    \BlankLine
    \For{$\zeta_1 \in C$}
    {
        \If{there exists $\zeta_2 \in q$ such that $\zeta_1 \sim \zeta_2$}
        {
            {$q \gets q  \cup \{\zeta_1\}$}\\
            {$C \gets C \setminus \{\zeta_1\}$}
        }
    }
    $Q \gets Q \, \cup \, \{q\}$
}

}
\end{algorithm}

\section{Comparison and Stability}
\label{sec:tvecc-stability}
In this section, we introduce a comparison measure between two one-parameter PL families and study the stability properties of the time-varying Euler characteristic curve (TV-ECC) under perturbations of the underlying scalar field.

\inlineheading{Comparison measure.}
We propose a distance measure between two generic one-parameter PL families and use it to compare two time-varying scalar fields. The distance between two families $\{f_t \}$ and $\{g_t \}$, $t \in I$, is defined as the aggregated  difference between the corresponding TV-ECCs,
{
\begin{equation}
\Delta(\{f_t\},\{g_t\}) := \int_I \int_{ \Real} \mid \TVECC{f_t}(s,t) - \TVECC{g_t} (s,t) \mid \mathrm{d}s\,\mathrm{d}t.    
\end{equation}
}
Restricting to a single time step, the induced distance between a pair of generic PL functions $f$ and $g$ is
\[
\Delta(f,g) \;:=\; \int_{\mathbb{R}}
\bigl\lvert \ecc{f}(s) - \ecc{g}(s) \bigr\rvert \,\mathrm{d}s.
\]

\inlineheading{Stability.}
We first derive a result that shows the stability under the metric $\Delta$ for two PL functions, and later extend it to time-varying scalar fields.
The Wasserstein stability ($W_1$) of $\Delta$ follows from the result of Dlotko and Gurnari~\cite{ecc}:
\begin{equation}
\label{eq:w1_stability}
\Delta(f,g) \leq \sum_{k} W_1(D_{f,k}, D_{g,k}),
\end{equation}
where $W_1(D_{f,k}, D_{g,k})$ is the $1$-Wasserstein distance between the $k$-persistent barcodes of $f$ and $g$. \revision{However, the $1$-Wasserstein distance between persistence diagrams of sublevel set filtrations of scalar functions is not stable with respect to the $L^\infty$ norm in general \cite{skraba2025wassersteinstabilitypersistencediagrams}. So, we establish the following stability result.}

\begin{theorem}
Given two generic PL functions $f$ and $g$ on a combinatorial manifold $\Mspace$ with $n$ vertices, 
\[ \Delta(f,g)\;\le\; n\,l\,\lVert f-g\rVert_{\infty},\] 
where $l$ is an upper bound on the size of the link of a vertex.
\end{theorem}

\begin{proof}
Consider the one-parameter family
\[
\{F_t\} \colon \Mspace \times [0,1] \longrightarrow \Real,\quad
F_t(x) = (1-t)\,f(x) + t\,g(x).
\]
We may assume that $\{F_t\}$ is generic; otherwise, it can be perturbed infinitesimally to ensure genericity.
Suppose $0=t_1<t_2<\cdots<t_{m}=1$ is a partition of $[0,1]$ such that each interval $[t_i,t_{i+1}]$ contains exactly one double point of the collection of curves $\{(t, F_t(x))\}_{x\in V}$.
Use \autoref{thm:localsumthm} to express $\ecc{F_t}(s)$ as
\[
\ecc{F_t}(s) = \sum_{x\in V} \localTVECC{x}(s,t).
\]
For a vertex $x$, let $\localTVECC{x}(t)$ denote $\lim_{s\to\infty}\localTVECC{x}(s,t)$.
Consider one of the finitely many double points occurring in the time interval $[t_i,t_{i+1}]$ for two vertices $v,w\in V$.
Then, the statement of \autoref{thm:betti-update} can be rephrased as
\begin{equation}\label{eulerconserve}
\localTVECC{v}(t_i) + \localTVECC{w}(t_i)
=
\localTVECC{v}(t_{i+1}) + \localTVECC{w}(t_{i+1}).
\end{equation}
Note that $\localTVECC{x}(t)$ remains unchanged for $x\notin\{v,w\}$ between $t_i$ and $t_{i+1}$.
So, the difference between the ECCs is
\begin{align*}
\ecc{F_{t_i}}(s) &- \ecc{F_{t_{i+1}}}(s) =\\
&\sum_{\substack{x\notin\{v,w\}\\ F_{t_i}(x) < F_{t_{i+1}}(x)}}
\mathbf{1}_{\{F_{t_i}(x)\le s < F_{t_{i+1}}(x)\}}
\,\localTVECC{x}(t_i)\\
- 
&\sum_{\substack{x\notin\{v,w\}\\ F_{t_i}(x) > F_{t_{i+1}}(x)}}
\mathbf{1}_{\{F_{t_{i+1}}(x)\le s < F_{t_i}(x)\}}
\,\localTVECC{x}(t_i)\\
+ 
&\localTVECC{v}(s,t_i)+\localTVECC{w}(s,t_i)
-\localTVECC{v}(s,t_{i+1})-\localTVECC{w}(s,t_{i+1}).
\end{align*}
Next, we establish an upper bound on the absolute value of the difference,
$\bigl\lvert \ecc{F_{t_i}}(s) - \ecc{F_{t_{i+1}}}(s)\bigr\rvert$, for a time interval $[t_i,t_{i+1}]$.
Assume, without loss of generality, that $f(v)>g(v)$ and let
\[
D(s) :=
\localTVECC{v}(s,t_i) + \localTVECC{w}(s,t_i) -\localTVECC{v}(s,t_{i+1}) - \localTVECC{w}(s,t_{i+1}).
\]
When $F_{t_i}(v)>F_{t_{i+1}}(w)$ (see \autoref{fig:crossingfirst}), applying \eqref{eulerconserve} gives
\[
D(s) =
\begin{cases}
0 & \text{if } F_{t_i}(v) \le s,\\
-\localTVECC{v}(t_i) & \text{if } F_{t_{i+1}}(w) \le s < F_{t_i}(v),\\
-\localTVECC{v}(t_i)+\localTVECC{w}(t_{i+1})
  & \text{if } F_{t_{i+1}}(v) \le s < F_{t_{i+1}}(w),\\
\localTVECC{w}(t_i) & \text{if } F_{t_i}(w) \le s < F_{t_{i+1}}(v),\\
0 & \text{if } s < F_{t_i}(w).
\end{cases}
\]
\begin{figure}[t]
\centering
\begin{tikzpicture}[
  x=1.05cm,y=0.85cm,
  font=\scriptsize,
  every node/.style={inner sep=1pt},
]

\def\xTi{0}
\def\xTip{4}

\def\yVti{3.2}   
\def\yWti{0.6}   
\def\yWtip{2.4}  
\def\yVtip{1.4}  

\coordinate (Vti)  at (\xTi,\yVti);
\coordinate (Vtip) at (\xTip,\yVtip);
\coordinate (Wti)  at (\xTi,\yWti);
\coordinate (Wtip) at (\xTip,\yWtip);

\path[name path=segV] (Vti) -- (Vtip);
\path[name path=segW] (Wti) -- (Wtip);
\path[name intersections={of=segV and segW, by=X}];

\draw[thick] (Vti) -- (Vtip);
\draw[thick] (Wti) -- (Wtip);

\fill (Vti)  circle (1.2pt)
  node[above left, align=left] {$(t_i,\ F_{t_i}(v))$};
\fill (Wti)  circle (1.2pt)
  node[below left, align=left] {$(t_i,\ F_{t_i}(w))$};

\fill (Wtip) circle (1.2pt)
  node[above right, align=left] {$(t_{i+1},\ F_{t_{i+1}}(w))$};
\fill (Vtip) circle (1.2pt)
  node[below right, align=left] {$(t_{i+1},\ F_{t_{i+1}}(v))$};

\fill (X) circle (1.4pt);

\draw (\xTi,0) ++(0,0.07) -- ++(0,-0.14) node[below] {$t_i$};
\draw (\xTip,0) ++(0,0.07) -- ++(0,-0.14) node[below] {$t_{i+1}$};

\path let \p1=(X) in coordinate (XT) at (\x1,0);
\draw (XT) ++(0,0.07) -- ++(0,-0.14) node[below, xshift=2pt] {$t$};

\end{tikzpicture}
\caption{Two vertices $v$ and $w$ crossing within a time interval $[t_i,t_{i+1}]$.}
\label{fig:crossingfirst}
\end{figure}
Similarly, when $F_{t_i}(v) < F_{t_{i+1}}(w)$,
\[
D(s) =
\begin{cases}
0 & \text{if } F_{t_{i+1}}(w) \le s,\\
\localTVECC{w}(t_{i+1}) & \text{if } F_{t_i}(v) \le s < F_{t_{i+1}}(w),\\
-\localTVECC{v}(t_i)+\localTVECC{w}(t_{i+1})
  & \text{if } F_{t_i}(w) \le s < F_{t_i}(v),\\
\localTVECC{v}(t_{i+1}) & \text{if } F_{t_{i+1}}(v) \le s < F_{t_i}(w),\\
0 & \text{if } s < F_{t_{i+1}}(v).
\end{cases}
\]
In both cases,
\[
\lvert D(s)\rvert \le
l\Bigl(
\mathbf{1}_{\{F_{t_{i+1}}(v) \le s < F_{t_i}(v)\}}
+
\mathbf{1}_{\{F_{t_i}(w) \le s < F_{t_{i+1}}(w)\}}
\Bigr).
\]
Therefore,
\begin{align*}
\lvert \ecc{f}(s) - \ecc{g}(s)\rvert
&\le
\sum_{i=1}^{m-1}
\bigl\lvert
\ecc{F_{t_{i+1}}}(s) - \ecc{F_{t_i}}(s)
\bigr\rvert
\\
&\le
\sum_{f(x)>g(x)}
l\,\mathbf{1}_{\{g(x)\le s < f(x)\}}
\\
&\quad
+
\sum_{g(x)>f(x)}
l\,\mathbf{1}_{\{f(x)\le s < g(x)\}} .
\end{align*}
Integrating the difference between the ECCs over $\Real$, we get
\[
\Delta(f,g) \le \sum_{x\in V}  l\, |f(x)-g(x)| \le n\,l\,\lVert f-g\rVert_{\infty}.
\]

\end{proof}

The above stability theorem shows that $\Delta$ is stable with respect to the infinity norm. As a consequence, we obtain the following stability theorem for the distance $\Delta(\{f_t\},\{g_t\})$ between \revision{time-varying} scalar fields.
\begin{theorem}
Let $\{f_t\},\{g_t\}$ be two generic PL families defined on a combinatorial manifold with $n$ vertices, and $l$ an upper bound on the size of the link of each vertex. Then
\[
\Delta\bigl(\{f_t\},\{g_t\}\bigr)
\;\le\;
n\,l\int_{I}\!\lVert f_t-g_t\rVert_{\infty}\,\mathrm{d}t .
\]
\end{theorem}

Intuitively, aggregating the difference between the $\orecc$ of two scalar fields via integration results in a measure of the sum of differences between critical values of corresponding topological events. When applied to time-varying fields, this aggregation captures the cumulative critical value shift in these events across the entire time domain.

Stability guarantees support the application of the distance measure to noisy time-varying data. Establishing the stability theorem for the distance measure $\Delta$ provides a guarantee of robustness thereby increasing the reliability of the use of the topological descriptor.
We leverage this property in our computational experiments to analyze the dynamics of a time-varying field from pairwise distances between Cerf diagrams computed over sliding time windows.
\revision{Finally, as noted by Li et al.~\cite{Li2025FlexTopology}, the dependence of the stability bound on the combinatorics of the domain ($n$) can be attributed to $\Delta$ being an $L^1$ type norm. We present results of an empirical study of the tightness of the upper bound in the supplementary material.}

\section{Experimental results}
\label{sec:results}
\begin{figure*}[t]
\centering

\subfloat[]{%
    \includegraphics[width=0.33\textwidth]{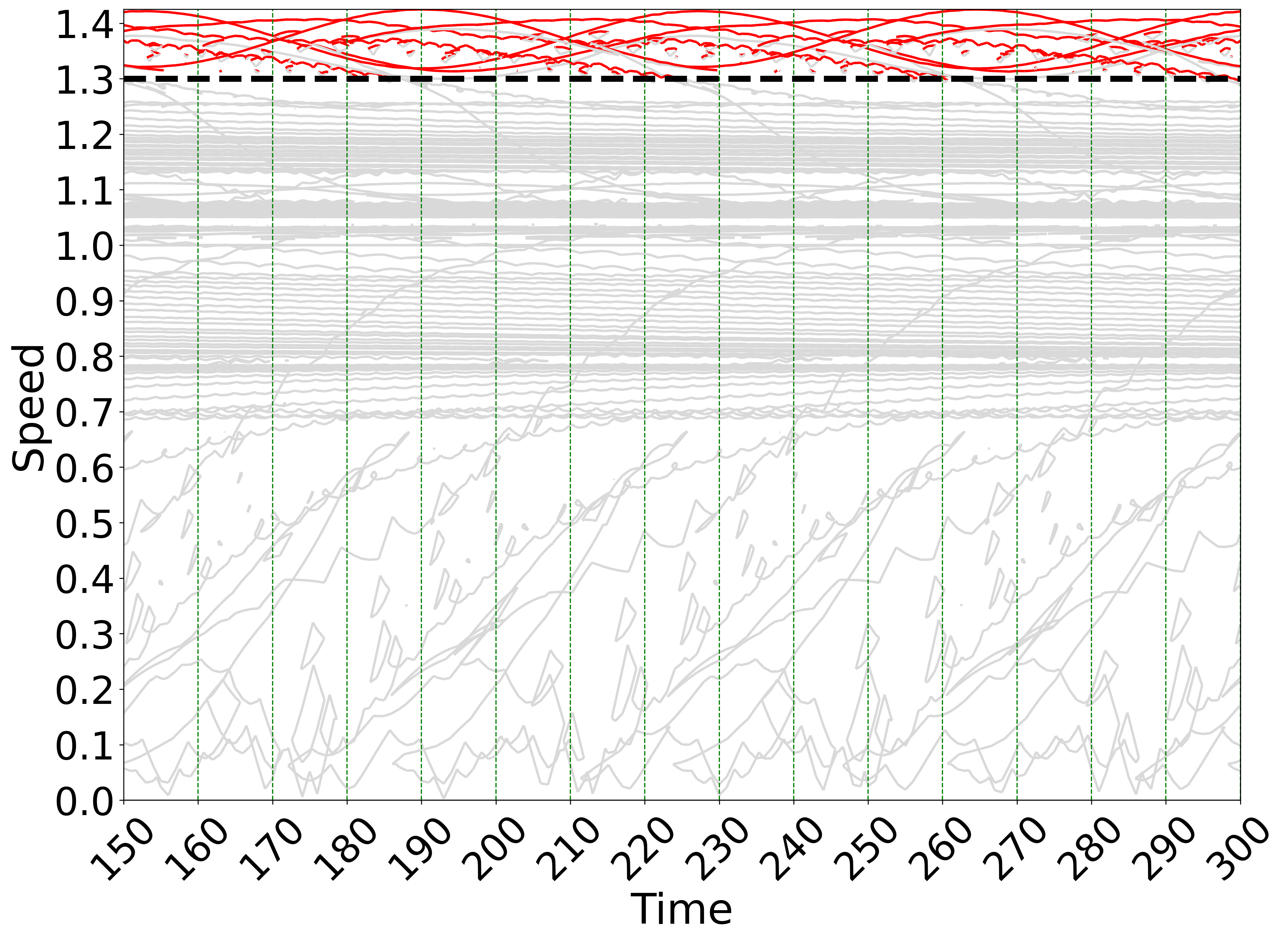}
    \label{fig: vstreet_a}
}
\subfloat[]{%
    \includegraphics[width=0.33\textwidth]{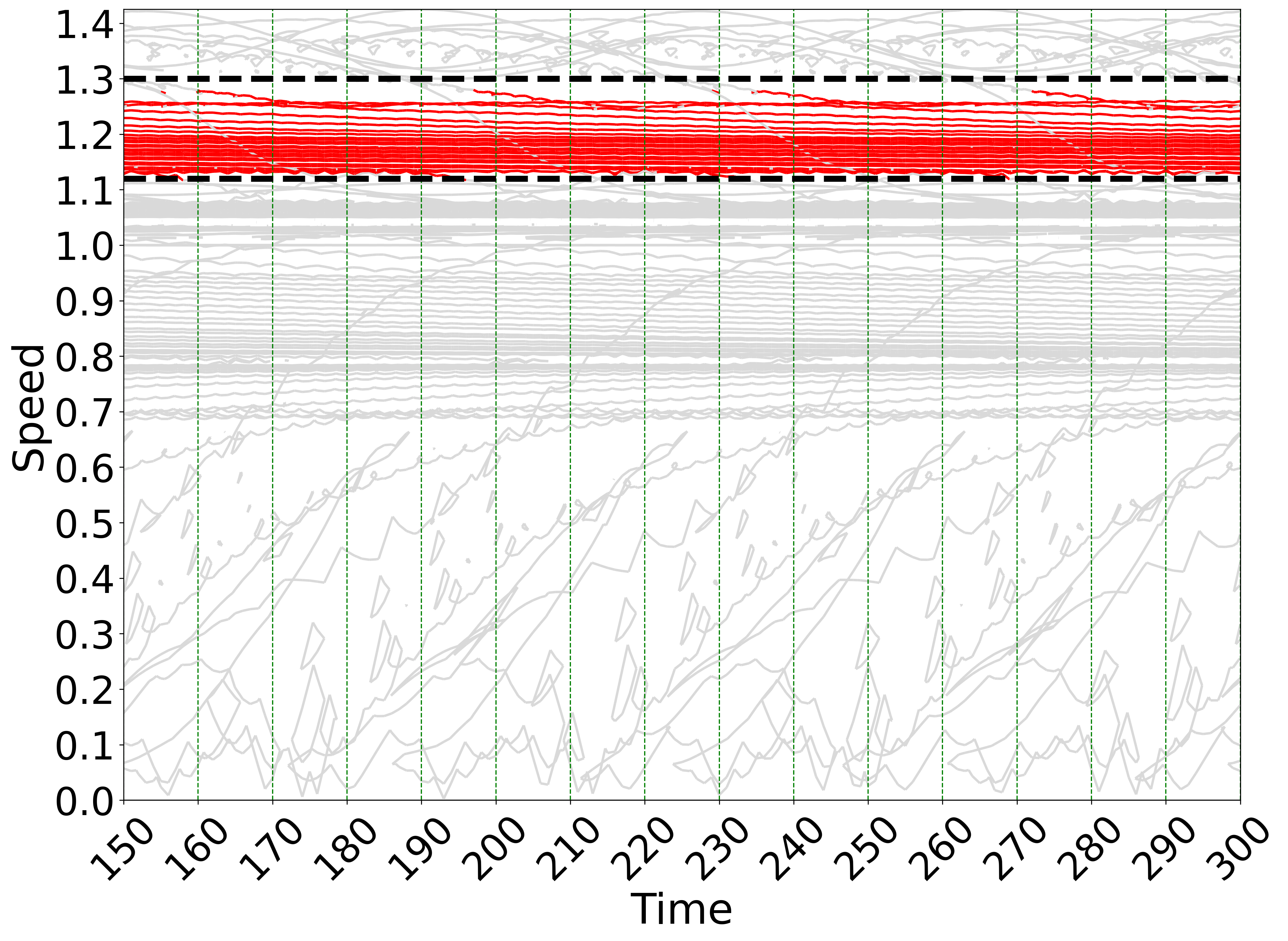}
    \label{fig: vstreet_b}
}
\subfloat[]{%
    \includegraphics[width=0.33\textwidth]{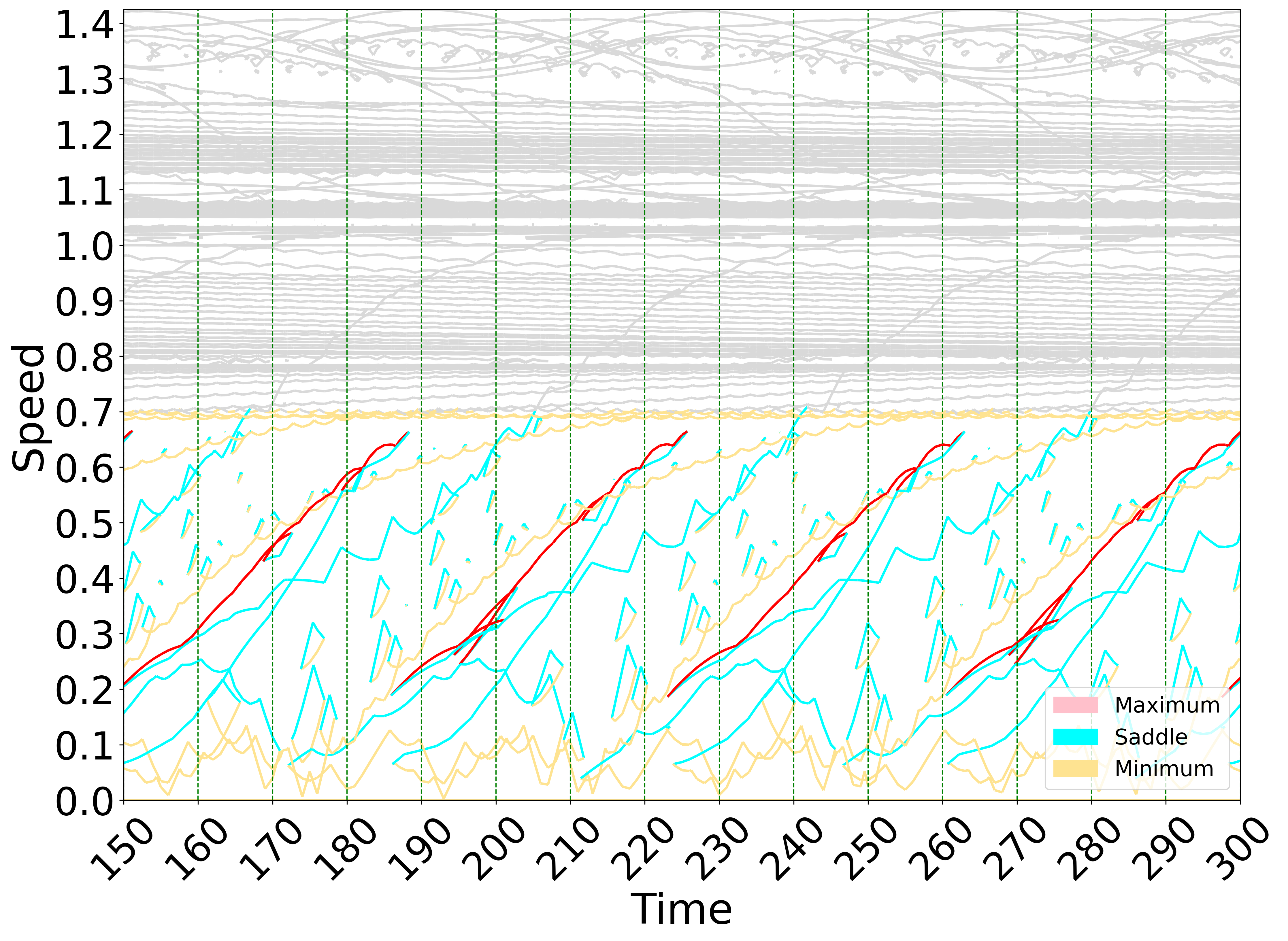}
    \label{fig: vstreet_c}
}
\par \vspace{2pt}
\subfloat[]{%
    \includegraphics[width=0.33\textwidth]{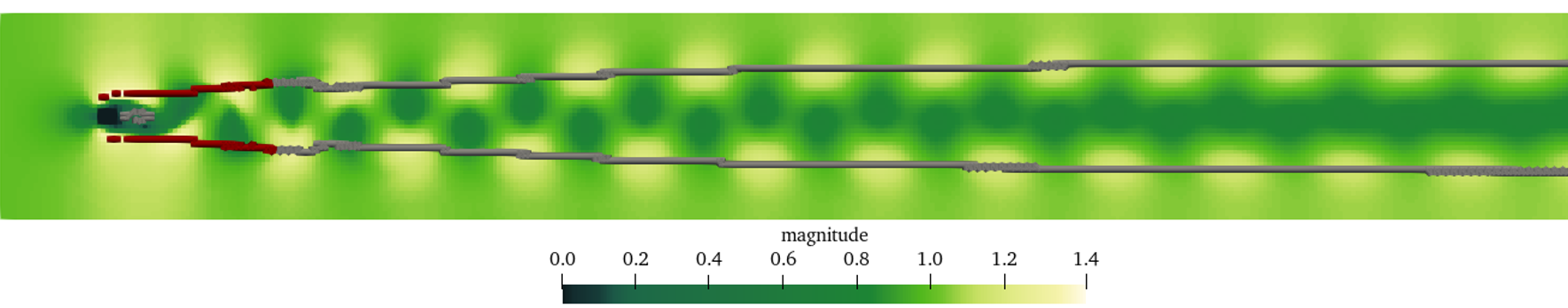}
    \label{fig: vstreet_d}
}
\subfloat[]{%
    \includegraphics[width=0.33\textwidth]{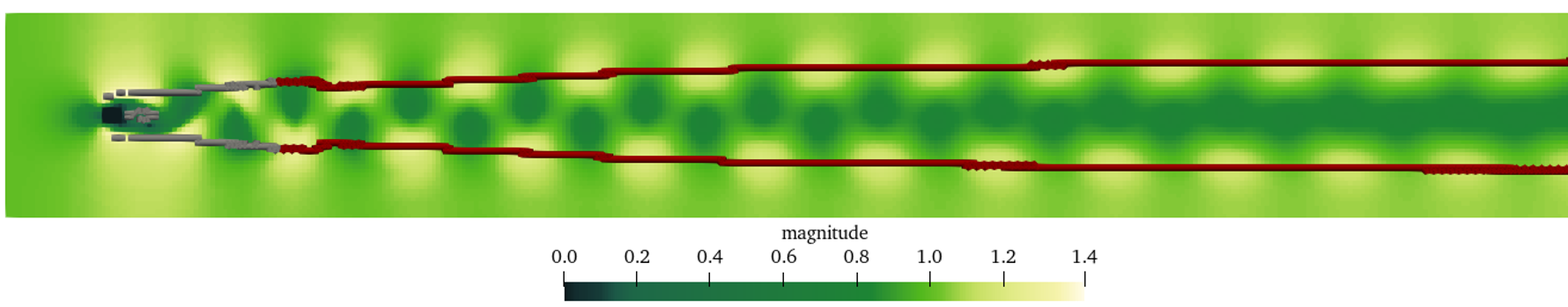}
    \label{fig: vstreet_e}
}
\subfloat[]{%
    \includegraphics[width=0.33\textwidth]{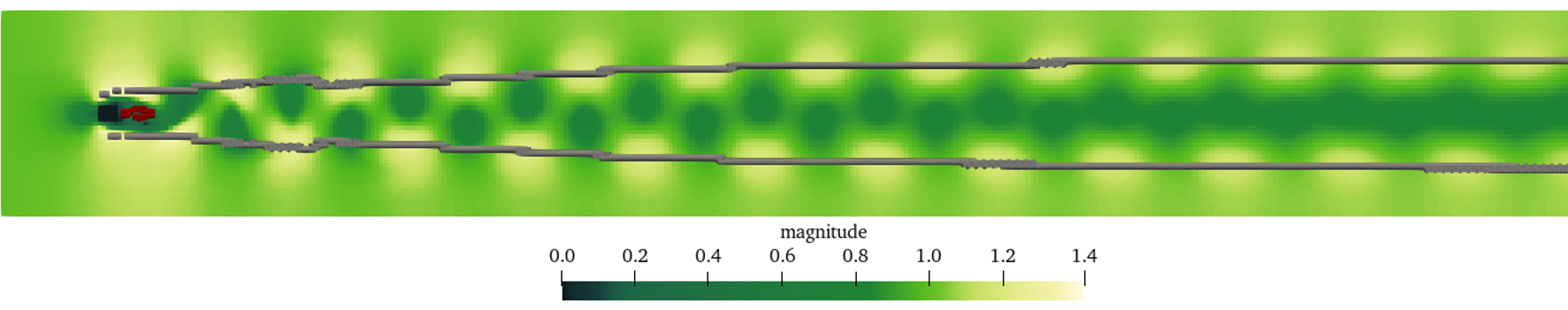}
    \label{fig: vstreet_f}
}
\par \vspace{2pt}
\begin{minipage}[c]{0.45\textwidth}
  \centering
\subfloat[]{%
    \includegraphics[width=\textwidth]{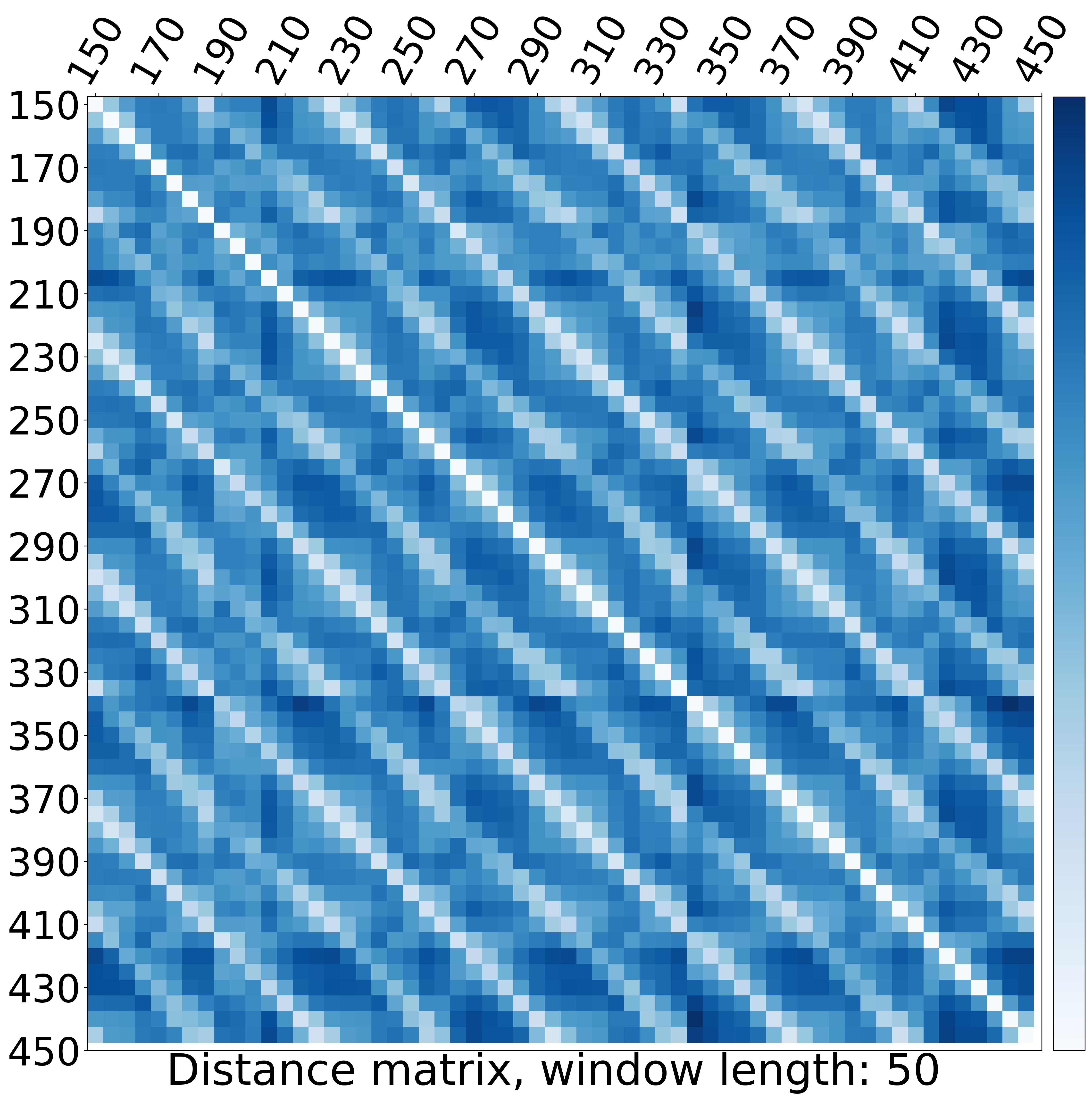}
    \label{fig: vstreet_g}
}
\end{minipage}
\hfill
\begin{minipage}[c]{0.45\textwidth}
  \centering
  \subfloat[]{
    \includegraphics[width=\textwidth]{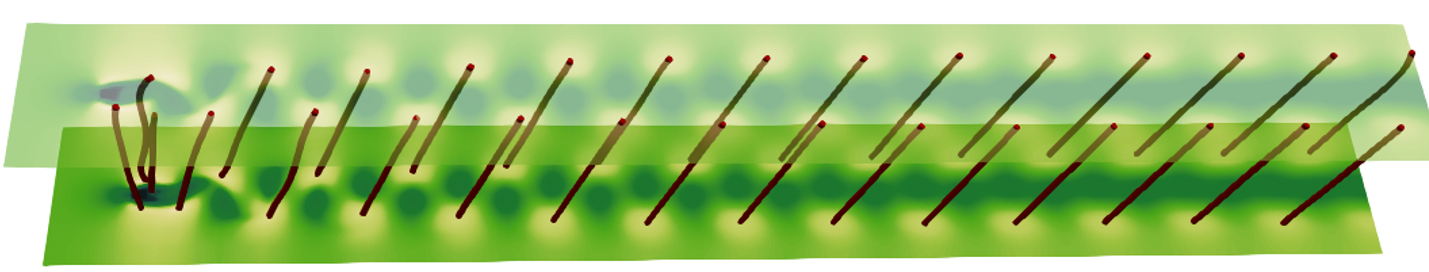}
    \label{fig: vstreet_h}
  }\\[0.5em]
  \subfloat[]{
    \includegraphics[width=\textwidth]{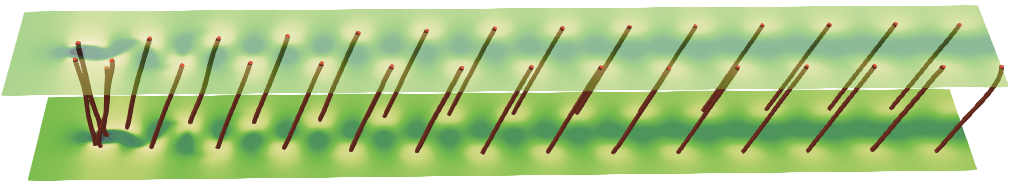}
    \label{fig: vstreet_i}
  }\\[0.5em]
  \subfloat[]{
    \includegraphics[width=\textwidth]{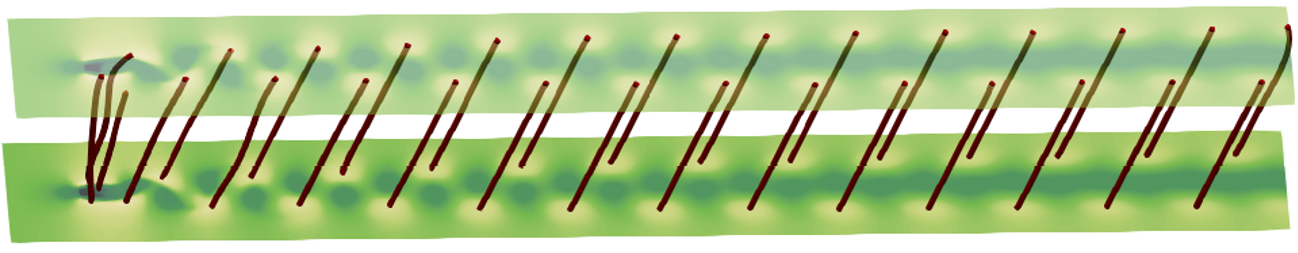}
    \label{fig: vstreet_j}
  }
\end{minipage}

\caption{Analyzing periodicity and types of topological features in the von K\'{a}rm\'{a}n vortex street dataset using the Cerf diagram.
Three distinct patterns of maxima can be identified from the Cerf diagram: a sinusoidal pattern~(a), flat pattern~(b), and a pattern of periodic maxima-saddle pairs~(c). Spatial tracks corresponding to the patterns highlight distinct topological features~(d,e,f).
Tracks of maxima in time windows of length 50 starting at $t = 185, 222, 270$~(h,i,j) reveal that the windows starting at $t = 185$ and $t = 222$ differ in terms of spatial tracks lying in the vicinity of the cylinder. The diagonal pattern in the matrix of pairwise distances between all pairs of time windows~(g) is indicative of periodicity with full period 75 and half period $\sim 37$.}
  \label{fig:2dVortexStreet}
\end{figure*}
\begin{figure*}[htb]
\centering
\subfloat[]{%
    \includegraphics[width=0.25\textwidth]{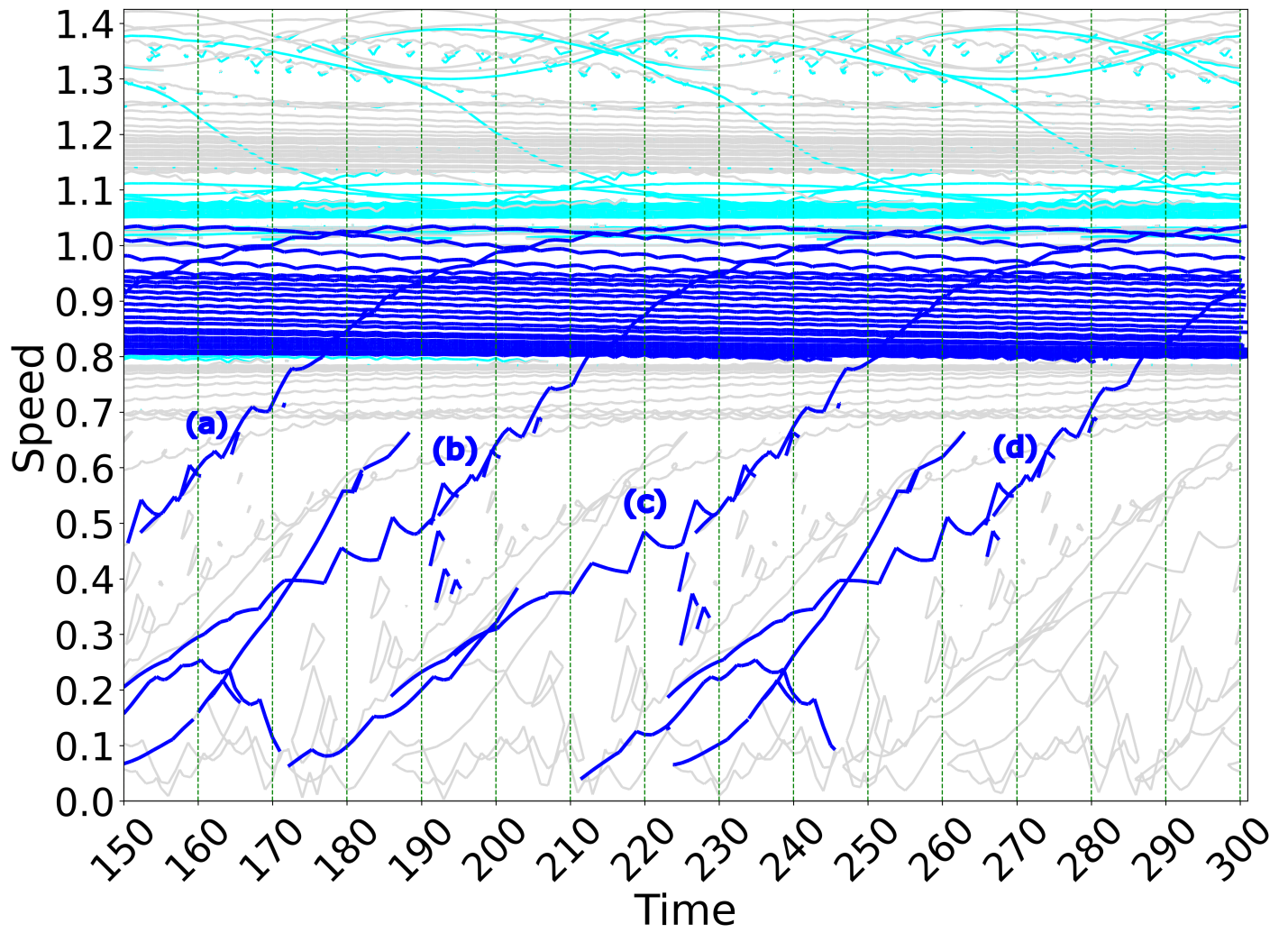}
    \label{fig:vstreetsaddle_a}
}\hspace{1in}
\subfloat[]{%
    \includegraphics[width=0.35\textwidth]{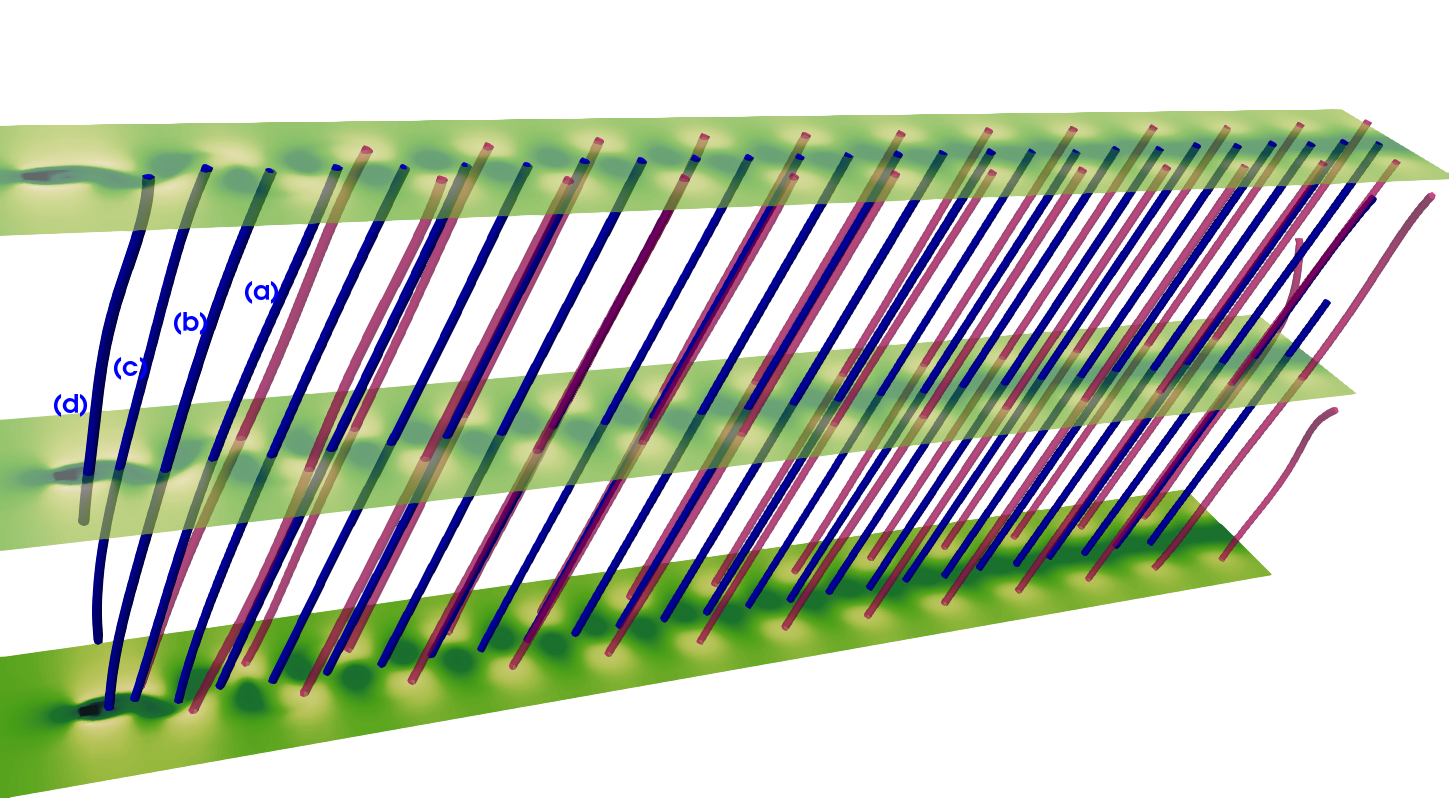}
    \label{fig:vstreetsaddle_b}
}\\
\subfloat[]{%
    \includegraphics[width=0.25\textwidth]{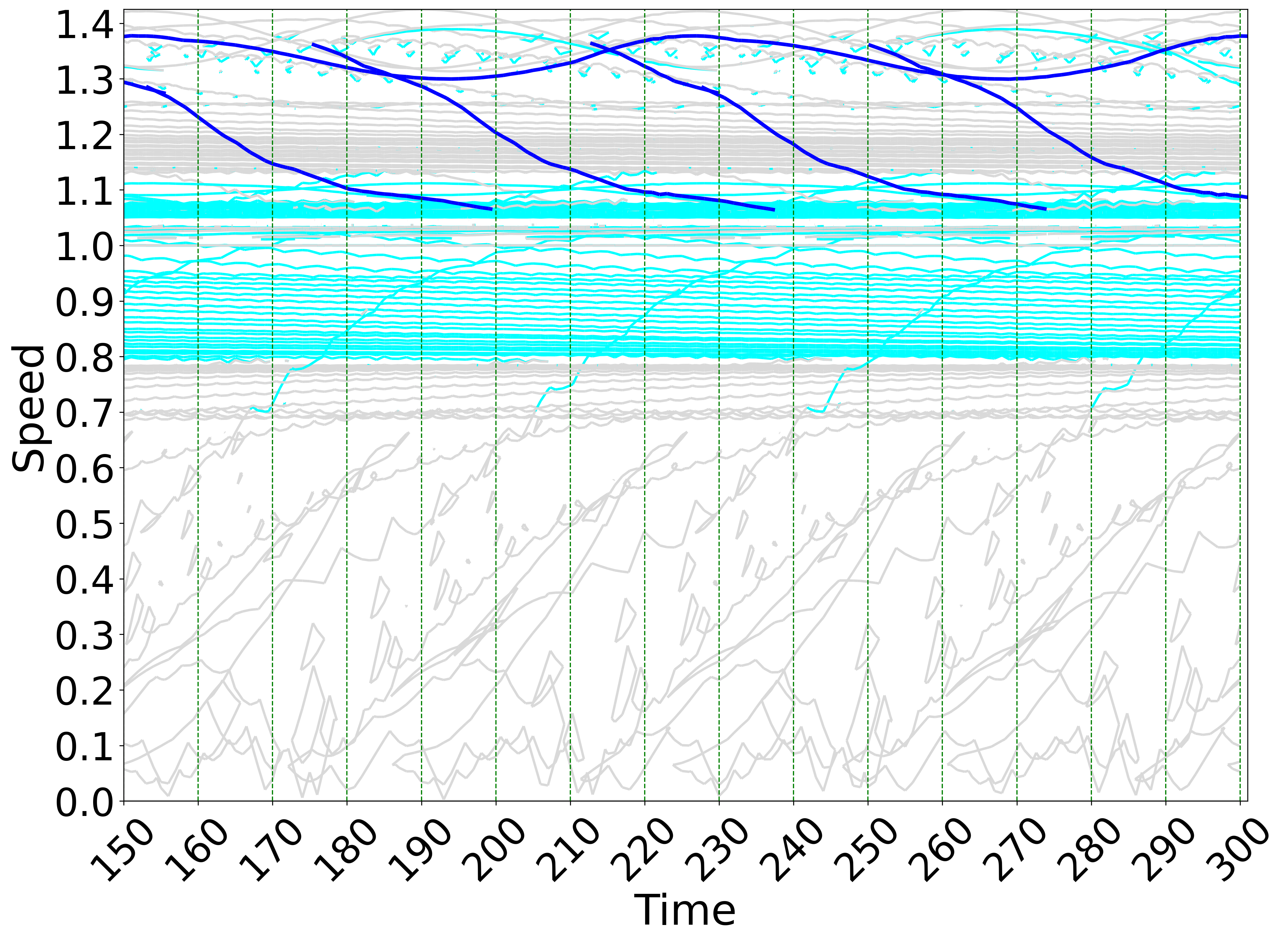}
    \label{fig:vstreetsaddle_c}
}\hspace{1in}
\subfloat[]{%
    \includegraphics[width=0.35\textwidth]{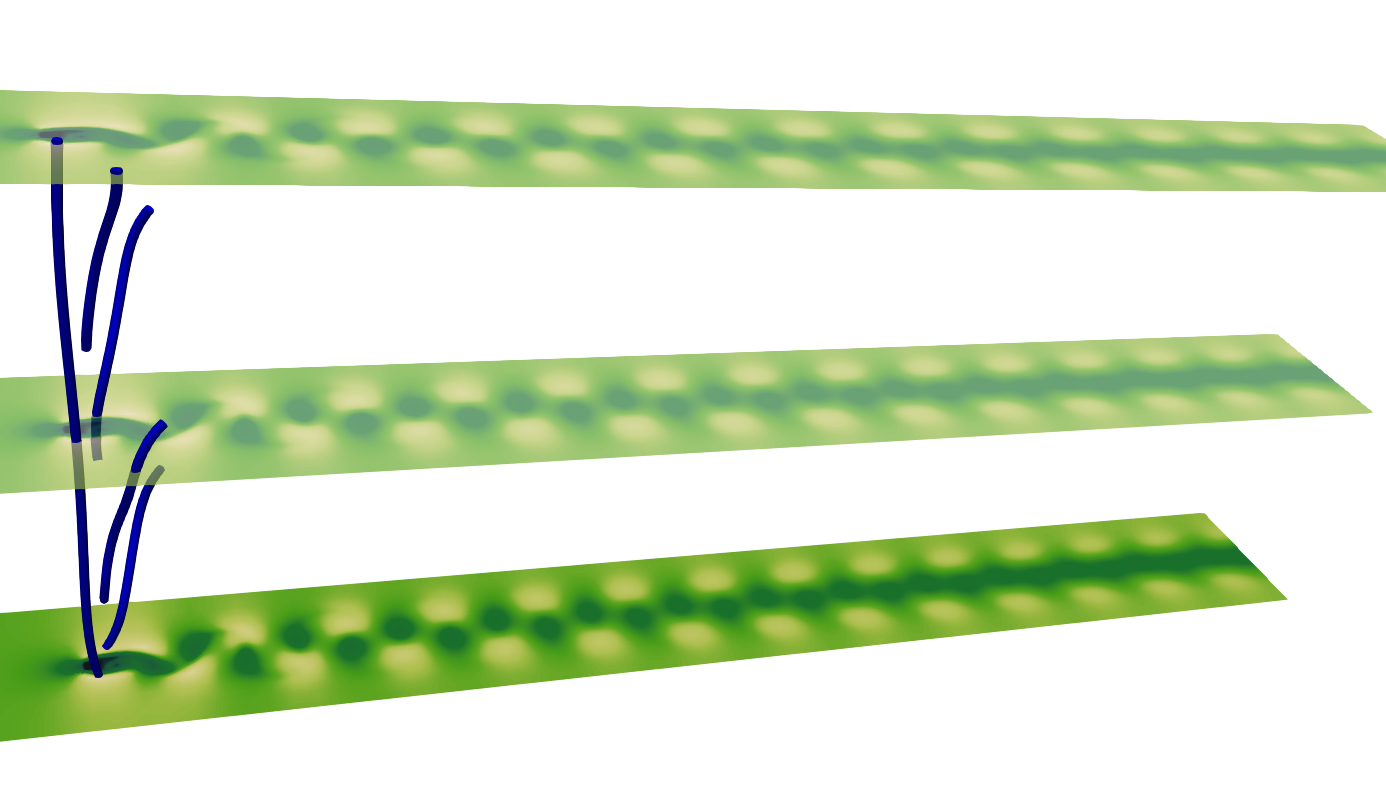}
    \label{fig:vstreetsaddle_d}
}
\caption{\revision{Querying Cerf diagram to study saddles within a time window [150,300].  (a)~ Blue Cerf tracks reported for a query with seed arcs in the range $[0.8,1.0]$. Other saddles with speed $\geq 0.8$ are shown in cyan. (b)~Corresponding spatial tracks. (c)~Blue Cerf tracks correspond to tracks originating above a speed value of 1.2. (d)~Corresponding spatial tracks. }
}
\label{fig:vstreetsaddle}
\end{figure*}
\begin{figure*}[!htb]
\centering
\subfloat[]{%
    \includegraphics[width=0.45\linewidth]{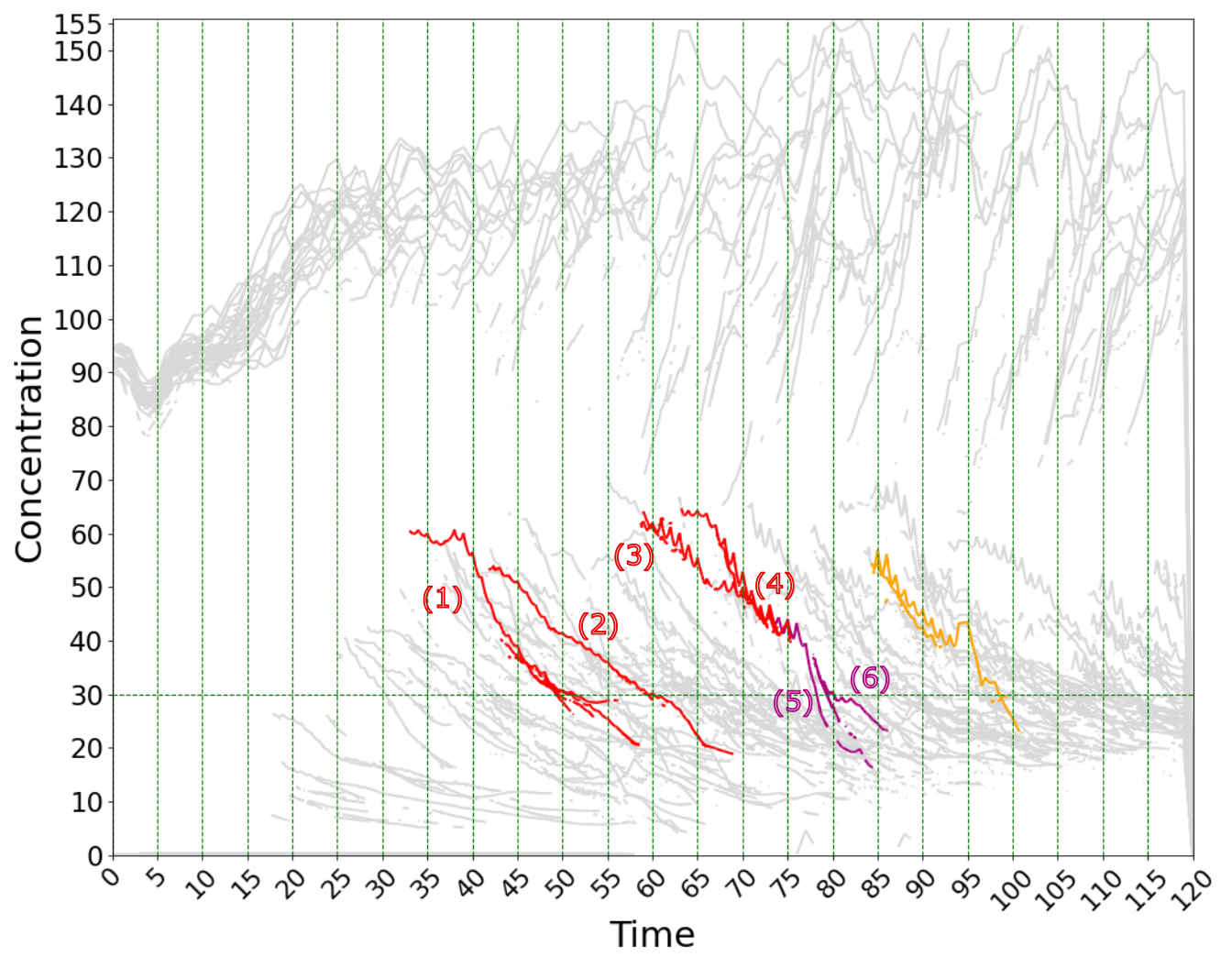}
    \label{fig:Fingers_cerfdiag}
}
\subfloat[]{%
\raisebox{2.0em}{%
    \includegraphics[width=0.28\linewidth]{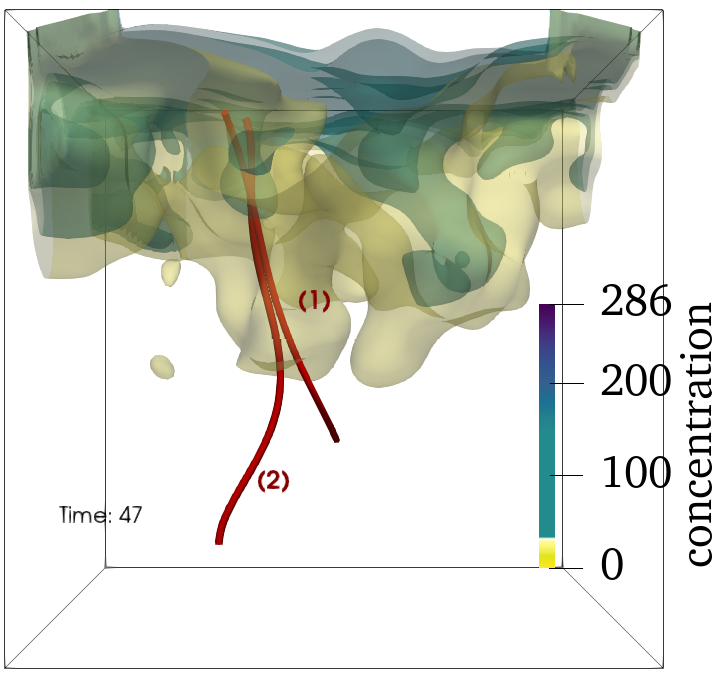}
}    \label{fig:Fingers_cerf_b}
}
\subfloat[]{%
\raisebox{1.9em}{%
    \includegraphics[width=0.265\linewidth]{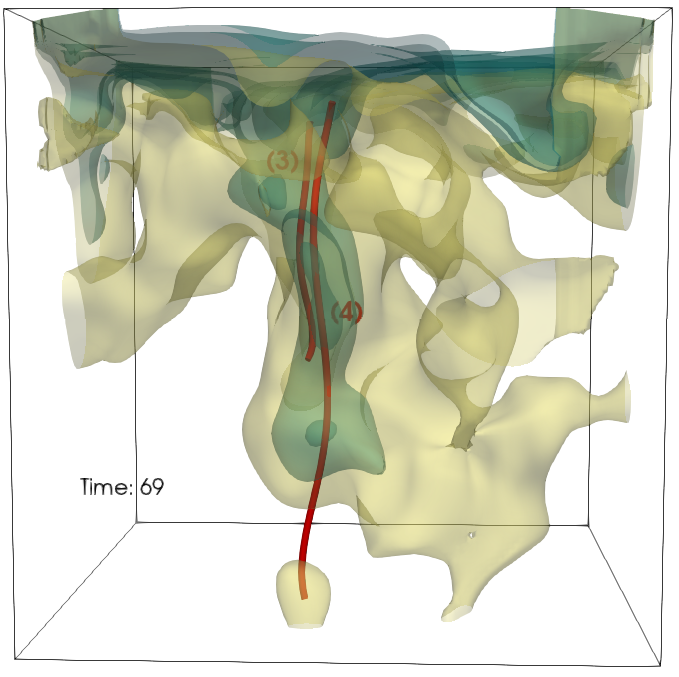}
}    \label{fig:Fingers_cerf_c}
}
\par \vspace{2pt}
\subfloat[]{%
    \includegraphics[width=0.247\linewidth]{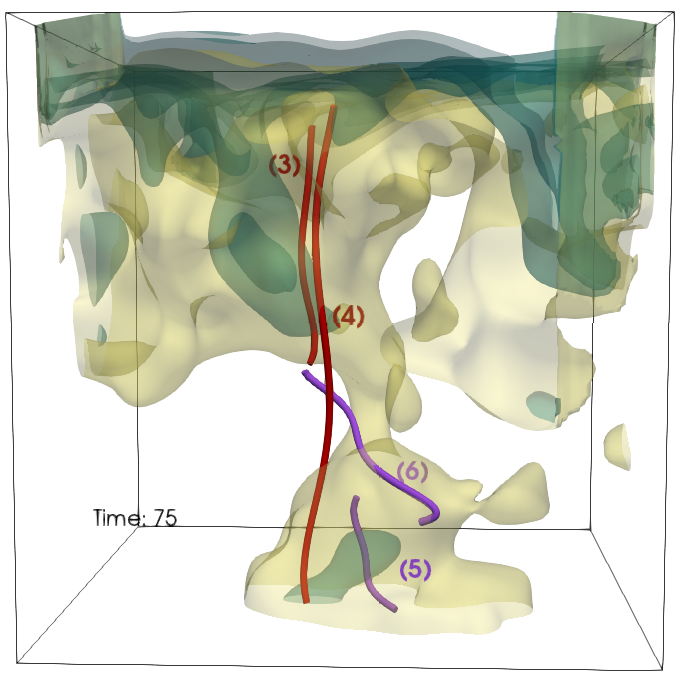}
    \label{fig:Fingers_cerf_d}
}
\subfloat[]{%
    \includegraphics[width=0.25\linewidth]{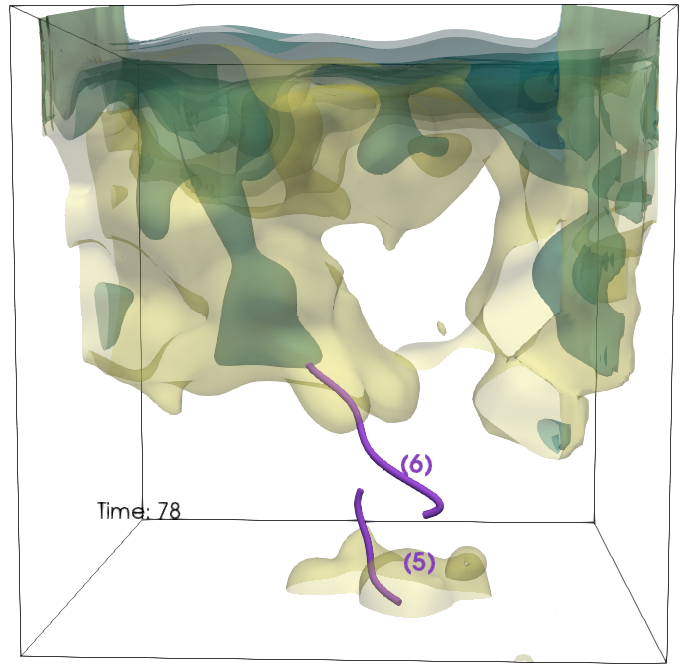}
    \label{fig:Fingers_cerf_e}
}
\subfloat[]{%
    \includegraphics[width=0.242\linewidth]{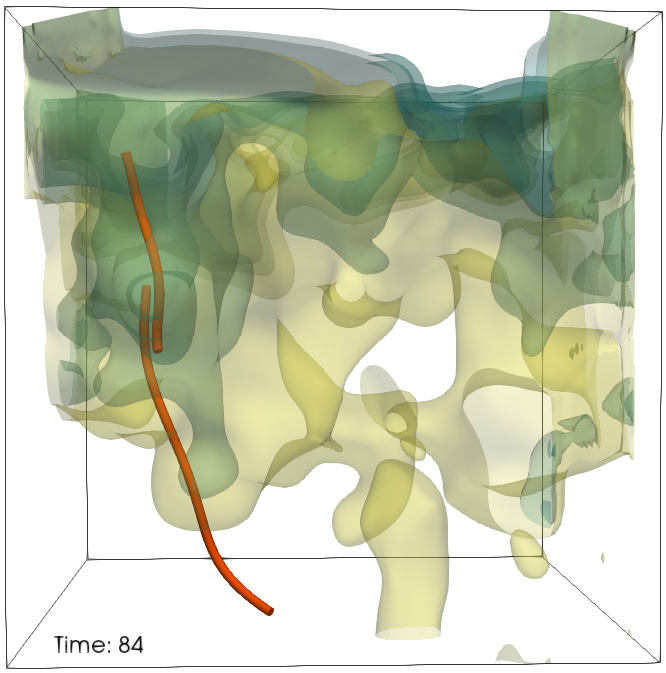}
    \label{fig:Fingers_cerf_f}
}
\subfloat[]{%
    \includegraphics[width=0.242\linewidth]{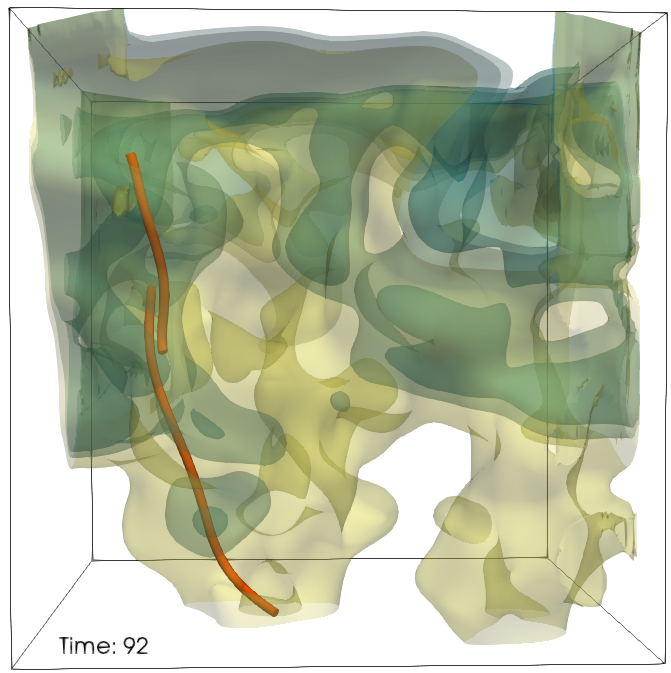}
    \label{fig:Fingers_cerf_g}
}

\caption{Studying the evolution of viscous fingers using the Cerf diagram~(a). The red Cerf tracks represent the formation of the central finger. These tracks terminate before $t = 75$, which is when the central finger breaks off~(b,c). Cerf arcs corresponding to maxima in the aftermath of the central finger breaking off are shown in purple~(d,e). Orange Cerf tracks, identified from the Cerf diagram based on their similarity with the previously studied tracks, represent auxiliary fingers~(f,g).}
\label{fig:Fingers_cerf}
\end{figure*}
\begin{table}[htb]
\centering
\caption{\revision{Running time for computing the Cerf diagram.}}
    \begin{tabular}{c|c|c}
    Data set & $\#$ Simplices & Running time (s) \\ \hline
    2D Vortex street& 0.15 Million & 20\\
    Viscous fingers & ~12 Million & 60\\
    Sea surface height & 0.2 Million & 25 
    \end{tabular}
\label{tab:runningtime}
\end{table}

We now report results of experiments on three real-world data, 2D and 3D time-varying scalar fields from the von K\'{a}rm\'{a}n vortex street, viscous fingers, and sea surface height datasets. The results indicate the potential utility of the Cerf diagram as a visualization descriptor that can summarize prominent time-varying topological features of interest. \revision{Computing the diagram takes less than a minute on a 20 core Intel Xeon(R) 6230 powered workstation with 384 GB RAM, as shown in \autoref{tab:runningtime}.} The synthetic example shown in \autoref{fig:cerf_2Gaussians} illustrates how the Cerf diagram provides a useful static overview of the time-varying field, depict periodicity in the field, and help identify salient time steps.

\subsection{2D vortex street}
The 2D von K\'{a}rm\'{a}n vortex street dataset is a simulation of a flow around a cylinder that exhibits periodic vortex shedding. The speed (velocity magnitude) is available as a scalar field over a $400\times 50$ grid across $1001$ time steps~\cite{weinkauf2010}. \autoref{fig:2dVortexStreet} shows the Cerf diagram computed over a time window $[150,300]$ and tracks of maxima in the spatial domain.
\revision{\autoref{fig: vstreet_g} shows the matrix of pairwise distances between Cerf diagrams corresponding to the restriction of the time-varying field to a sequence of time windows of length 50, beginning at time $t = 150$ with a shift of $5$.} 
We study the Cerf diagram with a focus on the maxima and saddles because they correspond to vortices, the primary feature of interest in this dataset.

\inlineheading{Identifying and classifying topological features.} 
We observe three types of maxima. First, maxima with speed greater than $1.3$ (\autoref{fig: vstreet_a}) that exhibit a sinusoidal pattern. They are located in the vicinity of the cylinder and highlighted in red in the spatial domain (\autoref{fig: vstreet_d}). Second, maxima with speed in the range $[1.12,1.3]$ (\autoref{fig: vstreet_b}) do not exhibit significant variation in the speed and lie further away from the cylinder (\autoref{fig: vstreet_e}). These maxima correspond to the vortices in the vortex street. Third, maxima with speeds lower than $1.12$ exhibit periodic birth-death behavior (\autoref{fig: vstreet_c}) with a lifetime of $\sim 37$ and occur near the cylinder (\autoref{fig: vstreet_f}).

\revision{\autoref{fig:vstreetsaddle_a} shows the Cerf tracks of saddles obtained by querying the Cerf diagram with a seeded search (speed $\in [0.8, 1]$). This query yields two patterns of Cerf tracks. The first pattern is due to saddles whose speed remains within the range $[0.8,1]$ in the time interval $[150,300]$. The corresponding spatial tracks (blue curves in \autoref{fig:vstreetsaddle_b}) occur away from the cylinder and between maxima lying at the center of the vortices (pink) tracks. The second pattern is four tracks (a-d) that originate at a speed $\sim 0.1$ but eventually flatten out near a speed of 1.0. Their spatial tracks indicate that they originate near the cylinder and eventually make their way between the vortex street maxima. These tracks are periodic and, between speeds 0.1 and 0.7, responsible for the periodic birth-death of maxima-saddle pairs reported in \autoref{fig: vstreet_f}.
A seeded query for speed $>$1.2 results in Cerf tracks (blue) shown in \autoref{fig:vstreetsaddle_c}. The sinusoidal shaped Cerf track lives within the entire interval $[150,300]$ (and continues beyond 300). Its spatial track remains close to the cylinder throughout its lifetime as shown in \autoref{fig:vstreetsaddle_d}.
We also observe a periodic pattern of saddles that are born after time intervals of $\sim 37$ (half period) at speed $> 1.3$ and disappear with speed $\sim 1.1$. Their (shorter) spatial tracks begin near the cylinder and move away subsequently.}

\inlineheading{Investigating periodicity and temporal events.} 
While the periodic nature of evolution of critical points may be directly observed in the Cerf diagram (region with speed $<0.7$), the repeating pattern is distinct in the distance matrix. Both the full period (75) and half period ($\sim 37$) can be deduced from the diagonal patterns in the matrix, and they match with previously reported results~\cite{saikia2014extended,extremumgraphs-narayanan,sridharamurthy2020}.

The above analysis was based on direct observation of patterns in the Cerf diagram. In subsequent experiments, we use the observations to pose relevant queries on the Cerf diagram, say extract Cerf tracks from a given collection of seeds. These tracks, prioritized based on their lifetimes and plotted as spatial tracks in the domain, correspond to important time-varying features highlighting the potential power of querying the Cerf diagram.

\subsection{Viscous fingers}
\begin{figure*}[!htb]
\centering
\includegraphics[width=0.35\linewidth]{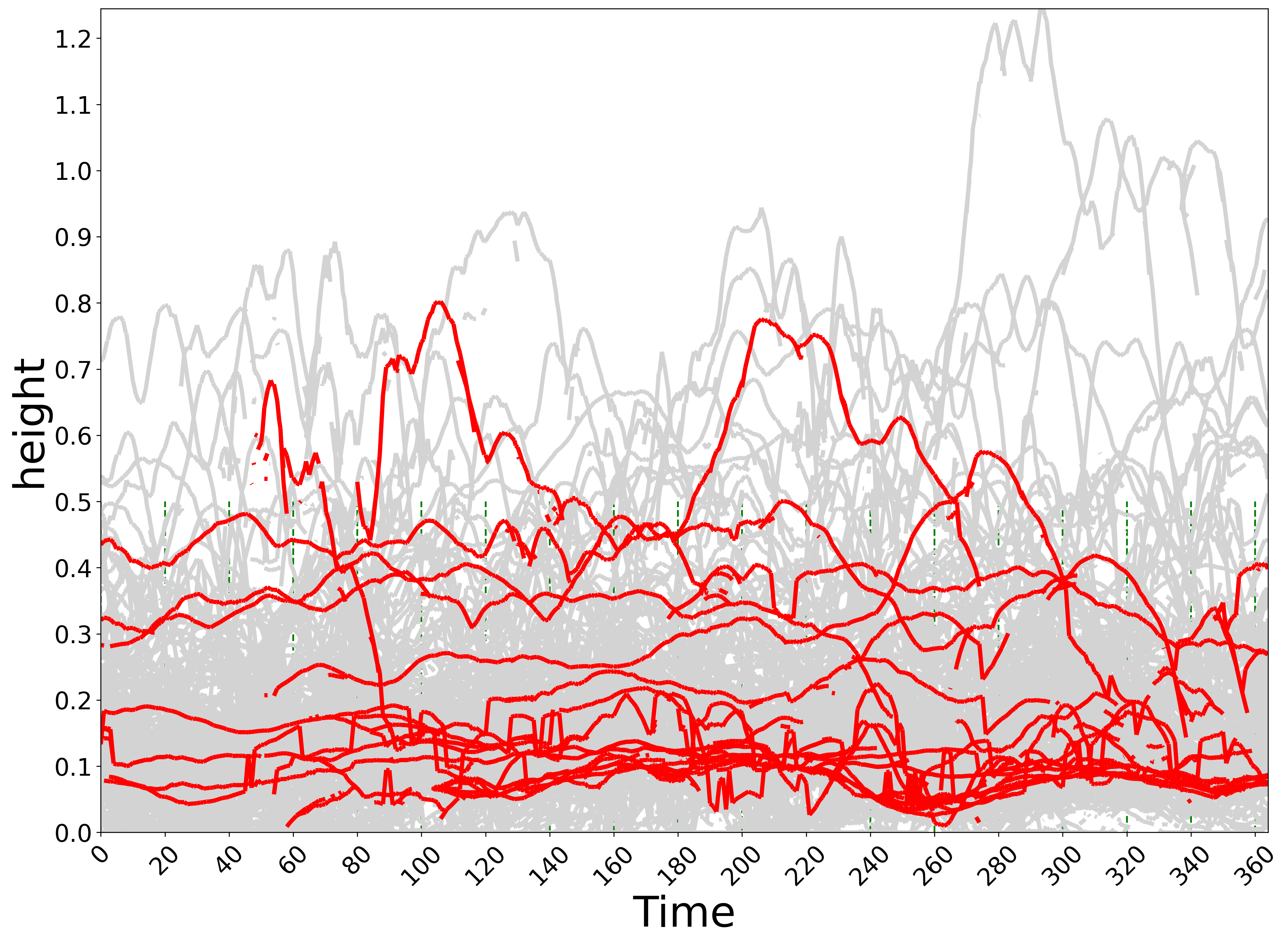} \hspace{0.1in}
\includegraphics[width=0.35\linewidth]{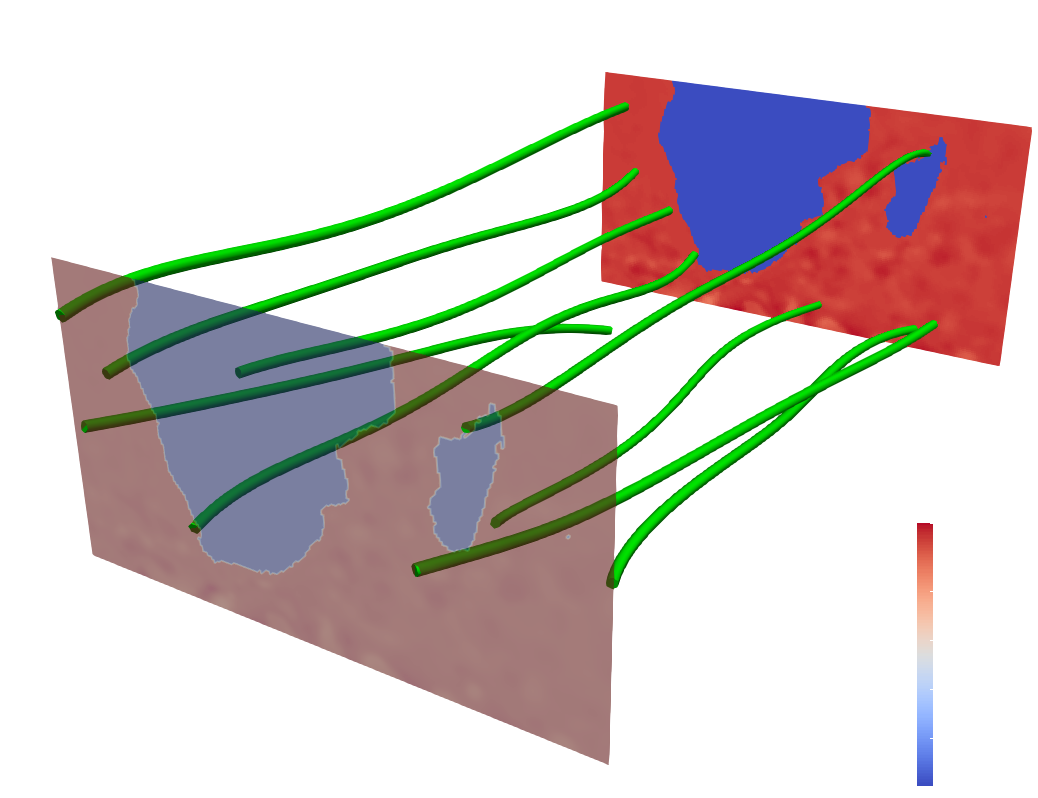}
\caption{Studying the evolution of sea surface height. (left)~A clustering heuristic followed by a lifetime based threshold of 75$\%$ results in the identification of 8 Cerf tracks. (right)~Corresponding spatial tracks rendered as degree-4 polynomial curves. Tracks that originate near the west coast of Africa show a distinct westwardly movement, as reported in previous studies.}
\label{fig:seasurface_cerf}
\end{figure*}
\begin{figure*}[!htb]
\centering
\includegraphics[width=0.38\linewidth]{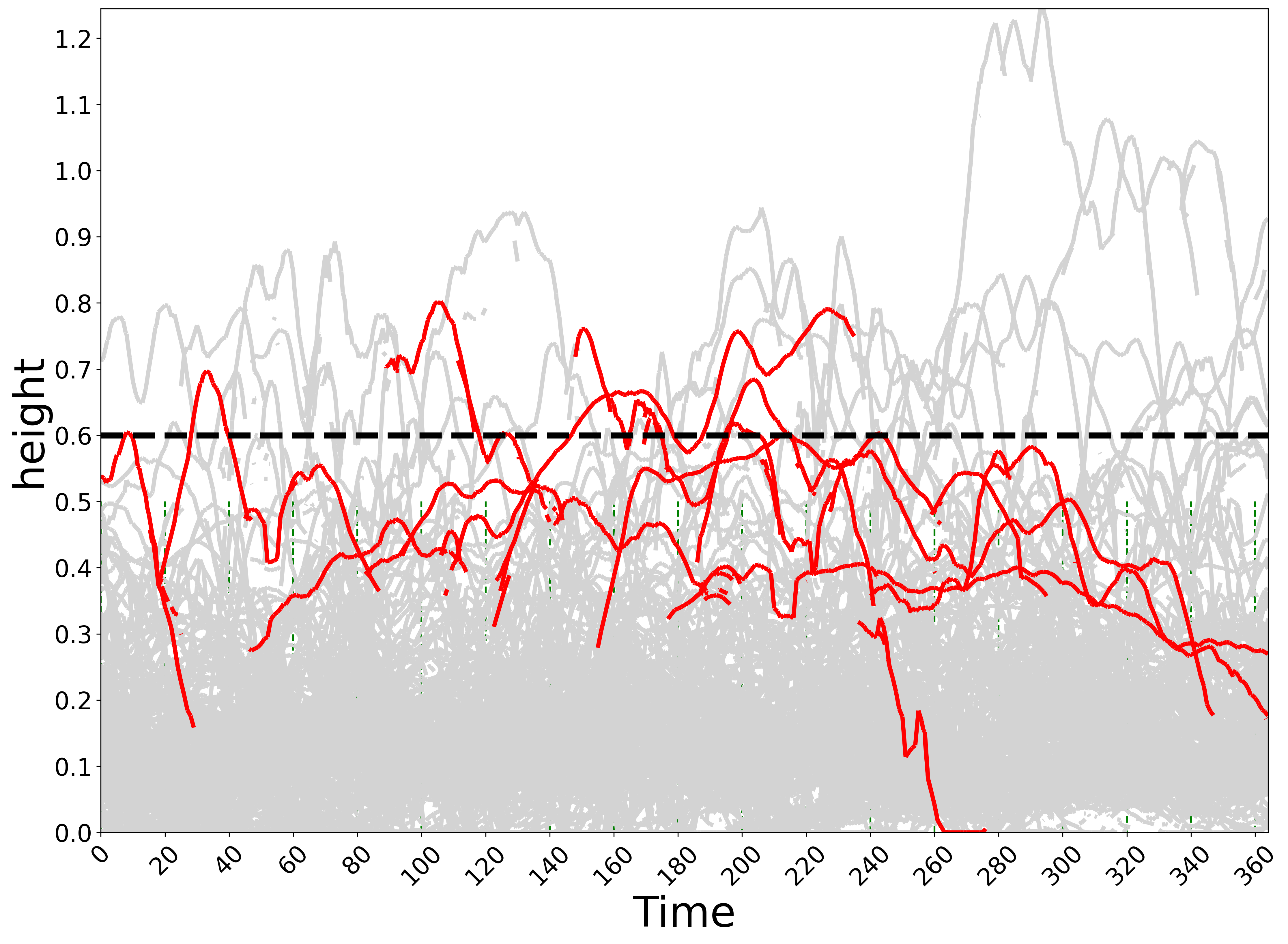}\hspace{0.1in}
\includegraphics[width=0.38\linewidth]{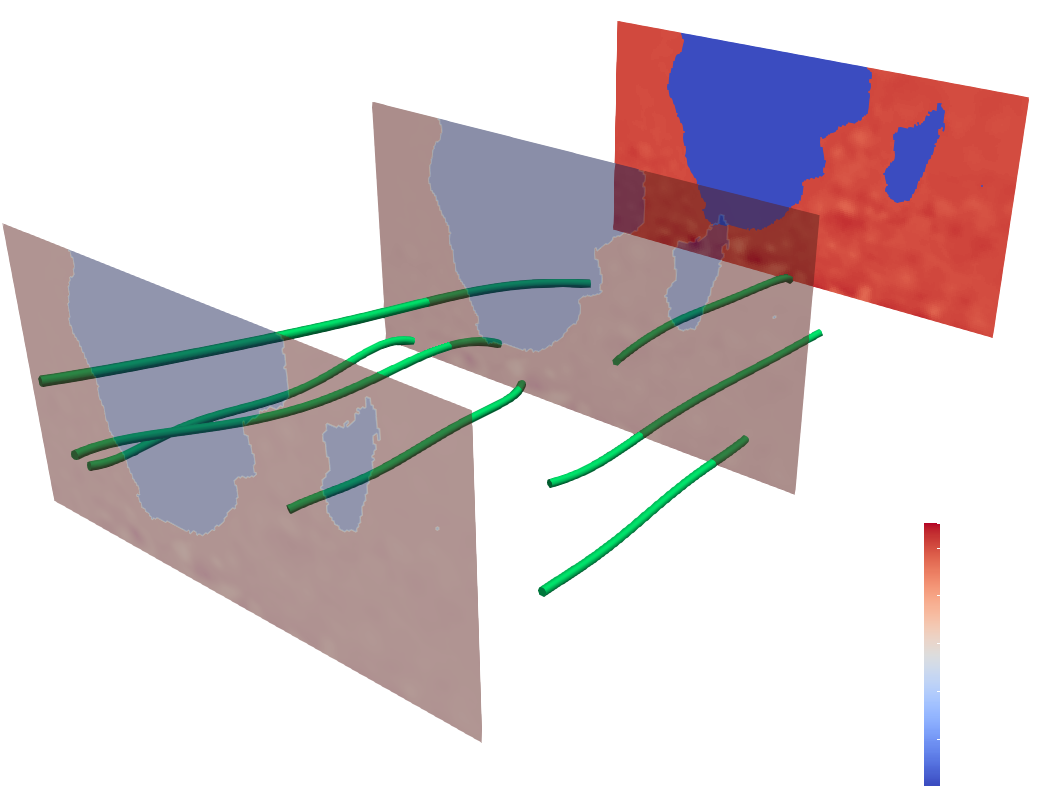}
\caption{Querying the Cerf diagram to study the evolution sea surface height. (left)~A query on the Cerf diagram with a seed set consisting of Cerf arcs where the maximum has a height value greater than 0.6 results in the identification of 7 Cerf tracks. (right)~Corresponding spatial tracks rendered as degree-4 polynomial curves. Tracks show a north-westward movement.}
\label{fig:seasurface_cerf_ht}
\end{figure*}
The viscous fingers dataset is the result of a stochastic simulation of the mixing of high concentration salt into water~\cite{viscousfingers2016dataset}. We analyze the $33^{rd}$ member of an ensemble at a smoothing length of 20 and resampled over a $101 \times 101 \times 101$ grid across 120 time steps~\cite{favelier2016visualizing}. The analysis begins with a study of maxima in the Cerf diagram shown in \autoref{fig:Fingers_cerf}.


\inlineheading{Identifying and classifying topological features.}
\autoref{fig:Fingers_cerfdiag} shows all Cerf arcs corresponding to maxima in the Cerf diagram. The upper band of Cerf arcs, with salt concentration greater than $70$, correspond to maxima that form near the mixing interface close to the upper boundary of the grid.
The lower band of Cerf arcs correspond to maxima located in the interior of the grid and is the focus of the analysis below.
We observe that maxima in the lower band with a concentration value exceeding 55 exhibit a steady decrease in concentration over time. These arcs correspond to the formation of fingers. 

\inlineheading{Querying the Cerf diagram.}
To verify the above hypothesis, we pose the following spatial query on the Cerf diagram: report Cerf tracks seeded at maxima that lie within the cylindrical region in the spatial domain centered at $(x,y) = (0.5,0.5)$ with a radius of $1$. The result of this query are the Cerf tracks shown in red in \autoref{fig:Fingers_cerfdiag}. \revision{Cerf tracks were computed using ($\Delta t = 2, \Delta f = 5\%, k = 2$).}

Indeed, each of the four reported Cerf tracks (1-4) corresponds to finger formation. Cerf tracks (1) and (2) originate at  $t=34$ and $t=41$, respectively.
A snapshot of the finger together with the spatial tracks is shown in \autoref{fig:Fingers_cerf_b} and in the supplementary video. The spatial track is represented by a degree-4 polynomial curve fit over the collection of critical points evolving over time.  The tracks begin near the mixing interface with concentration values exceeding 55 followed by a decrease in concentration as dissolution progresses.
Tracks (3) and (4) originate at  $t=58$ and $t=63$, and correspond to similar finger formation as seen in the spatial tracks (\autoref{fig:Fingers_cerf_c}). Track (4) disappears by  $t=75$, corresponding to the central finger breaking off as seen in \autoref{fig:Fingers_cerf_d}.
The query also yields tracks (5) and (6), shown in purple, whose behavior is distinct from the red tracks. These tracks originate after $t=74$ and disappear by $t=85$, thus having a much shorter lifespan than their red counterparts. They appear in the aftermath of the death of the central finger and do not correspond to distinct fingers, see \autoref{fig:Fingers_cerf_e}.

Further, we observe the presence of other Cerf tracks similar to the red tracks in the Cerf diagram, for example a couple of tracks originating at $t=85$. With the aim of determining if they represent similar features, we pose a query that results in the identification of the two orange tracks in  \autoref{fig:Fingers_cerfdiag}. These Cerf tracks terminate at $t=90$ and $t=100$, respectively. Extracting the corresponding spatial tracks, we observe that they represent the formation and evolution of two auxiliary fingers (\autoref{fig:Fingers_cerf_f}, \autoref{fig:Fingers_cerf_g}).

\inlineheading{Spatial and temporal queries.}
A Cerf arc records the vertex corresponding to the critical point, thereby supporting a search to report Cerf tracks based on seed vertices lying within a (spatial) query region. This type of query helped in the study of the central finger evolution.
The investigation of the auxiliary finger formation and evolution was supported by a temporal query, by specifying a temporal range ($t\in[85,90]$) and value range (salt concentration $\in [55,60]$).

\vspace{-0.1in}
\subsection{Sea surface height}

We now present the results of a Cerf diagram-based analysis of an oceanographic dataset from Copernicus Marine Service~\cite{copernicus}, where the sea surface height is recorded for 365 days from August 2021 to August 2022. The sea surface height scalar field is available on a $223 \times 116$ grid over 365 time steps.

\autoref{fig:seasurface_cerf} shows all Cerf tracks of maxima\revision{, computed with ($\Delta t = 2, \Delta f = 5\%, k = 2$)}. We highlight the tracks whose lifetime exceeds 75\% of the number of time steps. 
The corresponding spatial tracks that originate near the west coast of Africa show a westward movement as reported in previous studies~\cite{soler2018lifted}. 
Next, we seek to investigate the behavior of maxima with a prominent sea surface height. We therefore pose a seeded query on the Cerf diagram, with a seed set consisting of Cerf arcs whose height is greater that 0.6.
\autoref{fig:seasurface_cerf_ht} highlights the reported Cerf tracks. The corresponding spatial tracks exhibit a stronger north-westward movement. Further studies and analysis are required to determine the relevance of this observation. 

A Cerf track-based exploration aims to identify and cluster features based on their proximity, both in the Cerf diagram and within the spatial domain.
\revisionsecond{
For the listed experiments, choosing $(\Delta t = 1, k = 1)$, leads to significantly shorter trajectories and fails to fully capture salient features. Allowing for a greater tolerance for the parameter values may lead to the clustering of otherwise dissimilar features.
}
\revisionsecond{
Parameter values for $\Delta t$ and $k$ may be set based on temporal and spatial scale of the phenomenon in comparison with the coarseness of the sample in time and space. For the above experiments, ($\Delta t = 2, k = 2$) served as a good default but if the input data is available at finer temporal resolution then a higher value of $\Delta t$ may also be chosen. Similarly, if the spatial domain is represented by a uniform grid, $k$ can be chosen such that the spatial tolerance is a fraction of the total volume of the grid. Finally, $\Delta f$ is determined based on the estimated noise in the scalar field. 
}

\section{Discussion}
\label{sec:discussion}
Topological methods have been applied in the past towards analysis of time-varying scalar fields. The focus in these methods was often the identification and tracking of structural characteristics represented by isosurfaces, critical points, or contours (connected components of a specific isosurface). 
Vineyards and Nested Tracking Graphs are two descriptors closely related to the Cerf diagram. While these methods are practically useful, a study of the theoretical foundations aids both in a better understanding of the methods and of potential generalizations or future extensions. Below, we present a comparison against these descriptors with the objective of an improved understanding of the roles that each descriptor may play while analyzing time-varying scalar fields. We also summarize methods for critical point tracking and describe the role of the Cerf diagram in the study of spatio-temporal behavior of critical points. 


\subsection{Vineyard}
The \emph{vineyard} is a stack of persistence diagrams ordered by a parameter that captures the evolution of birth-death pairs in the persistence diagram~\cite{cohen2006vines}. Each point in the persistence diagram traces a path, called a \emph{vine}, in a 3D space determined by the birth and death time, and the parameter value. In the context of scalar fields and a filtration determined by lower stars, the persistence diagram consists of critical point pairs that represent the birth and death of a homology class. A vine represents the evolution of a persistent critical point pair and hence the corresponding homological feature. The vines are typically smooth except when the critical point pairing changes. 
For a PL Morse family, vines that begin or end at the diagonal of the persistence diagram represent birth or death events. In addition, a change in critical point pair owing to handle slides is also detected by vineyards. In comparison, a birth or death event can be directly read from the Cerf diagram but handle slides are not apparent.

The Cerf diagram captures the evolution of individual critical points as opposed to critical point pairs in the vineyard. Analogous to the vineyard enabling a study of the temporal evolution of the persistence diagram, the Cerf diagram supports the study of the temporal evolution of the $\orecc$. One advantage of the Cerf diagram is the ease of visualization, given that it is a 2D descriptor as opposed to the vineyard which is a 3D descriptor. A related advantage is the cheaper computational effort.

The vineyard algorithm processes a swap, a transposition of adjacent simplices in a filtration, in $O(m)$ time where $m$ is the total number of simplices. Transforming the lower star filtration at a given time step into that of the subsequent time step requires at most $O(m^2)$ swaps. So, the algorithm incurs a total cost of $O(m^3)$ for processing a pair of time steps. In contrast, our proposed algorithm takes only $O(nl^{\omega + 1})$ time in the worst case to compute the Cerf diagram between two consecutive time steps. 

\revision{
\subsection{Jacobi set}
The Jacobi set $\mathbb{J}(f,g)$ of a pair of scalar fields $f$ and $g$ is the set of solutions of $df + \lambda dg = 0$  or $ dg + \lambda  dg = 0$ for $\lambda \in \Real$~\cite{edelsbrunner2004jacobi,yan2021scalar}. Given a time-varying generic smooth family $f_t : M \to \mathbb{R}$, $t \in I = [0,1]$, consider a bivariate field $F = (f_t,t) : M \times I \rightarrow \mathbb{R}^2$. The Cerf diagram is equal to the image under $F$ of the Jacobi set $\mathbb{J}(f_t,t)$. Despite this correspondence, neither object determines the other. The Jacobi set is a locus that lives in the domain of the field $f_t$. The Cerf diagram, in contrast, lives in the time-function value plane, where it records the indices of the critical points and the evolution of their values. Jacobi sets admit a natural PL formulation, and we provide one for Cerf diagrams in this paper. In the context of analysis of time-varying scalar fields, Jacobi sets are especially useful for feature tracking and time-varying Reeb graph computations~\cite{ edelsbrunner2008timevaryingreeb}, while Cerf diagrams are helpful in a broader setting as they capture the global topological change induced by a time-varying family. The Jacobi is sensitive to noise and sparsity of samples due to it dependence on the gradient unlike the Cerf diagram which is characterized by the crossings. Methods for Jacobi set simplification are typically restricted to 2D fields~\cite{klotzl2022local}. 
}

\subsection{Nested tracking graph}
A tracking graph of topological and/or geometric features is a unified representation of temporal tracks of all features, where a node indicates one of the following events: birth, death, merge, split, or continuation. Arcs of the graph represent the correspondence between features in consecutive time steps, typically determined via spatial overlap. The \emph{nested tracking graph} consists of multiple tracking graphs and a nesting hierarchy between the nodes of the tracking graphs~\cite{Lukasczyk2020NestedTrackingGraph}. Each node of the nested tracking graph is associated with a specific time and scalar value. Nodes that belong to the same time step form a nesting tree where an arc between two nodes represents the nesting hierarchy between the superlevel set (or sublevel set) components represented by the two nodes. Arcs connecting nodes that are associated with the same scalar value form the tracking graph for the corresponding level set.

The nested tracking graph is applicable to scalar functions defined on non-manifold domains and hence more general as compared to the Cerf diagram. However, extensions to higher dimensions is impractical due to the computational effort. In comparison, the Cerf diagram is a topological descriptor for scalar field in higher dimensions. A successful use of the nested tracking graph depends on a careful choice of a finite number of scalar values or levels. The Cerf diagram does not require such an input. In the case of time-varying PL fields defined on triangulated manifolds, the Cerf diagram can supplement the nested tracking graph resulting in an enhancement  both in terms of quality of representation and computational efficiency: 
\begin{itemize}[leftmargin=*]
\item \textit{Choice of levels}: The Cerf diagram provides an overview of the evolution of maxima both spatially and temporally. This can help inform the choice of levels in the nested tracking graph, as opposed to the persistence diagram or split tree~\cite{Lukasczyk2020NestedTrackingGraph} that is restricted to an individual time step.
\item \textit{Choice of time steps}: The Cerf diagram likewise may guide the choice of time steps or time intervals for a particular application. For example, in a dataset that exhibits periodicity, time steps may be chosen to align with the time period.
\item \textit{Correspondence and tracks}: The temporal evolution of maxima can help determine arcs and tracks as an alternative to the spatial overlap criterion.
\item \textit{Computational efficiency}: The absence of crossings of index-($d-1$) critical points within a time interval in the Cerf diagram is indicative of no changes in the nesting tree. This may aid in improving the efficiency of the nested tracking graph computation.
\end{itemize}


\subsection{Critical point tracking}
Several methods have been proposed for tracking critical points over time. While effective in practice and being topology-aware, these methods are typically based on heuristics.
The critical points in a given time step may be represented using a topological descriptor such as the persistence diagram or the merge tree. A popular approach towards tracking determines correspondence between critical points in consecutive time steps by computing an optimal matching between the respective persistence diagrams~\cite{soler2018lifted}. A match with a diagonal point of the persistence diagram indicates the birth or death of a critical point track. In contrast to computing a matching between topological descriptors, this paper directly analyzes the underlying time-varying scalar field by establishing a PL analog of Morse-Cerf theory.

The lifted Wasserstein matcher~\cite{soler2018lifted} uses a weighted average of the $L_p$ distance between points in the persistence diagrams and geometric proximity of the associated critical points to identify an optimal match. A recently introduced probabilistic topology tracking method~\cite{Li2025FlexTopology} uses a fused Gromov Wasserstein distance between merge trees, which are modeled as attributed measure networks. While the tracking methodology may be extended to support the visualization of the critical values as well, this method and related variants~\cite{yan2021scalar,Yan2023} are limited to tracking only minima and 1-saddles (or maxima and ($d-1$)-saddles). Further, the computations often rely on parameters such as the weight~\cite{soler2018lifted} or a probability measure~\cite{Li2025FlexTopology} that is determined based on the application.
Computing the persistence diagram requires $O(n^\omega)$ time, an optimal matching using the Kuhn-Munkres algorithm requires $O(n^3)$ time, resulting in $O(n^3T)$ time for tracking all critical points over $T$ time steps. Optimal transport remains a computational bottleneck for the merge tree based matching.

\subsection{Limitations}
The approach described in this paper does not incorporate the above-mentioned critical point tracking methods. While the Cerf diagram does support a method for tracking topological features, the aim is to present a framework for systematic study of the spatial and temporal evolution of critical points in time-varying PL fields. Further, the study is limited to one-parameter PL Morse families. This is justified because the PL interpolation represents a simple approach towards establishing a continuous time-varying scalar field from a temporal sample. The Cerf diagram relies on the characterization of crossings in the PL vertex diagram, and hence may not capture features that are missed due to a sparse temporal sample from an otherwise unknown continuous time-varying scalar field.
\revisionsecond{
Characterizing the crossings is a performance bottleneck, particularly for large data sizes with finely triangulated 3D domains. A pre-simplification of the domain and the associated scalar field may result in improved running times.
}

The Cerf diagram suffers from visual clutter due to the presence of potentially insignificant spatio-temporal features. User-query and Cerf track driven exploration may mitigate this issue, but it remains a limitation of our proposed method for time-varying scalar field analysis. Effective methods to simplify the Cerf diagram, possibly using PL Cerf moves, are required to improve utility.

\section{Conclusions}
We introduced a PL adaptation of Morse-Cerf theory with the objective of supporting topological analysis of a one-parameter family of PL functions. We described a structural characterization of a topological descriptor called the Cerf diagram, methods for computing and comparing Cerf diagrams, stability results for the comparison measure, and applications to time-varying scalar fields to demonstrate its potential use for the analysis of time-varying scalar fields. 
Future work includes a comprehensive study of Cerf moves in the PL setting, methods for topological simplification, enhanced queries on the Cerf diagram\revision{, and a direct computation of select Cerf tracks to improve runtime performance.}

\section*{Acknowledgments}{%
This work is partially supported by the PMRF, MoE Govt. of India, a SERB grant CRG/2021/005278, and a CSR grant from Ittiam Systems for the Equitable AI Lab. VN acknowledges support from the Alexander von Humboldt Foundation and Berlin MATH+ under the Visiting Scholar program. 
}

\bibliographystyle{IEEEtran}
\bibliography{references}

\vspace{-0.1in}
\section{Biography Section}

\vspace{-0.49in}
\begin{IEEEbiography}[{\includegraphics[width=1in,height=1.25in,clip,keepaspectratio]{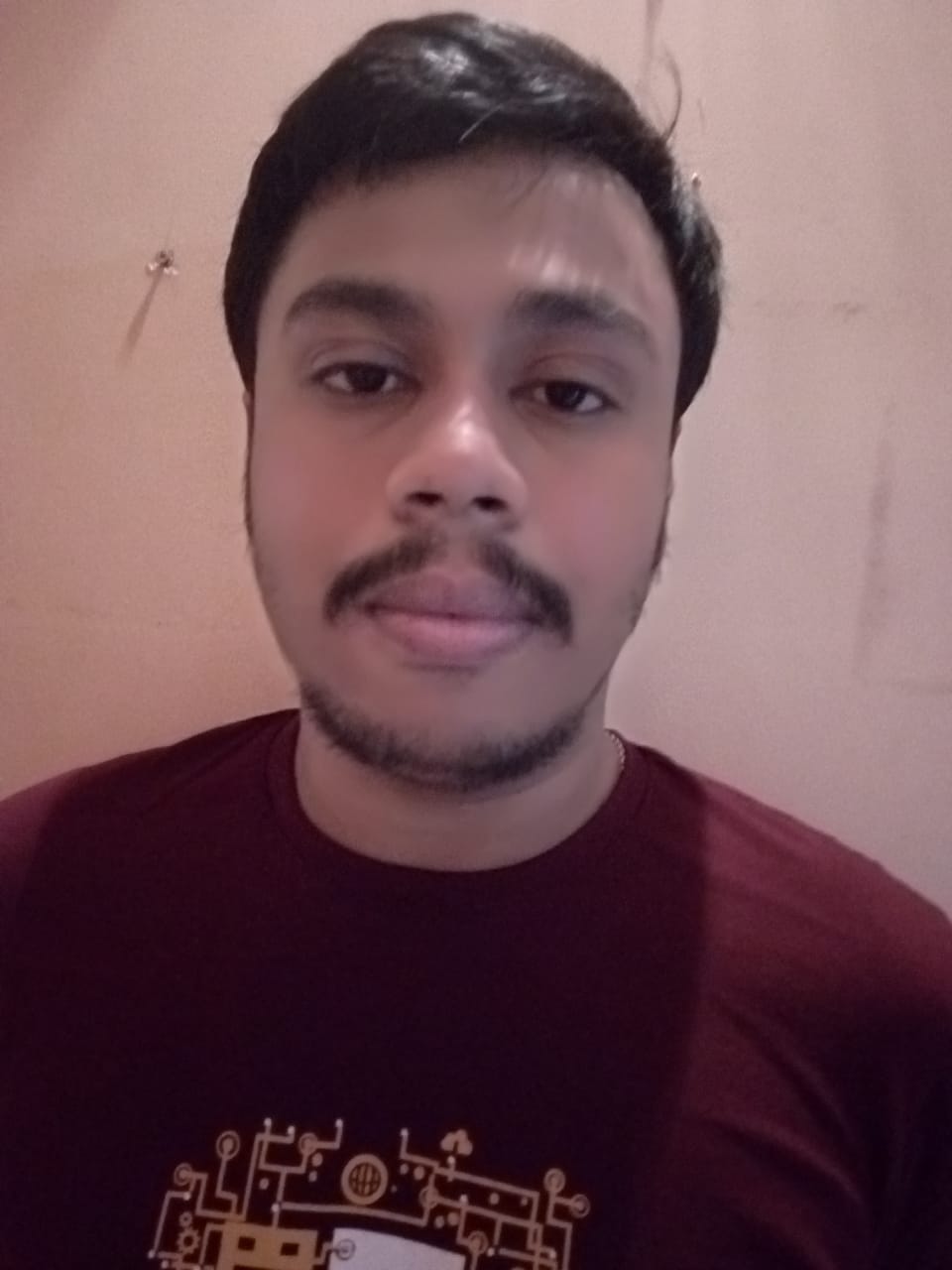}}]{Amritendu Dhar} is a Ph.D. student in the department of Computer Science and Automation at the Indian Institute of Science, Bangalore. He received his B. Tech. degree in Computer Science and Engineering from NIT Warangal and spent about three years at Microsoft IDC as a software developer. His scientific interests lie in the theoretical and algorithmic aspects of Computational Topology and Topological Data Analysis. Currently he is working on extending Morse theoretic descriptors to time-varying fields.
\end{IEEEbiography}

\vspace{-0.49in}
\begin{IEEEbiography}[{\includegraphics[width=1in,height=1.25in,clip,keepaspectratio]{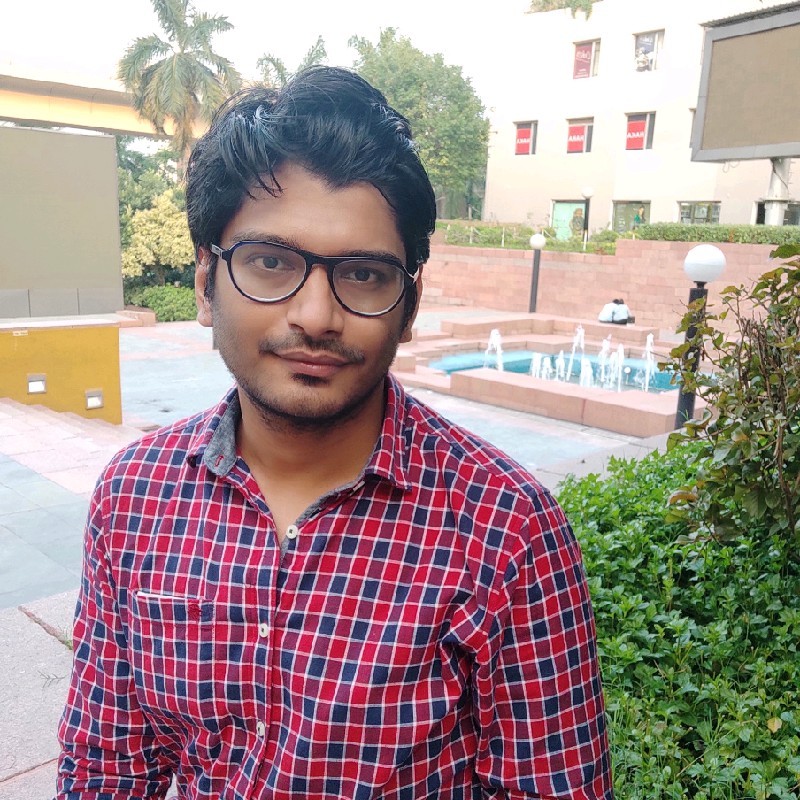}}]{Apratim Chakraborty} is an Assistant Professor at the Institute for Advancing Intelligence (IAI), TCG CREST. He received the Ph.D. degree in mathematics from Stony Brook University. His research interests include low-dimensional topology—specifically Heegaard Floer homology and quantum link invariants—and applied topology. In current work, he is developing categorified equivariant invariants and exploring  applications in topological combinatorics and topological data analysis.
\end{IEEEbiography}

\vspace{-0.49in}
\begin{IEEEbiography}[{\includegraphics[width=1in,height=1.25in,clip,keepaspectratio]{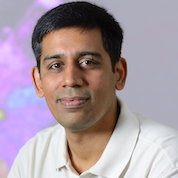}}]{Vijay Natarajan} is a Professor in the Department of Computer Science and Automation at Indian Institute of Science, Bangalore. He received the Ph.D. degree in computer science from Duke University. His research interests include scientific visualization, computational topology, and topological data analysis. In current work, he is developing topological methods for time-varying and multi-field data visualization, and studying applications in material science and climate science.
\end{IEEEbiography}

\vfill

\end{document}